\documentclass[12pt,reqno]{amsart}
\usepackage{amsthm,amsfonts,amssymb,euscript}

\newcommand{\bea}{\begin{eqnarray}}
\newcommand{\eea}{\end{eqnarray}}
\def\beaa{\begin{eqnarray*}}
\def\eeaa{\end{eqnarray*}}
\def\ba{\begin{array}}
\def\ea{\end{array}}
\def\be#1{\begin{equation} \label{#1}}
\def \eeq{\end{equation}}

\newcommand{\nn}{\nonumber}

\def\a{{\alpha}}

\def\b{{\beta}}
\def\be{{\beta}}
\def\ga{\gamma}
\def\Ga{\Gamma}
\def\de{\delta}
\def\De{\Delta}
\def\ep{\epsilon}

\def\la{\lambda}

\def\si{\sigma}
\def\Si{\Sigma}
\def\om{\omega}

\def\th{\theta}

\def\varep{\varepsilon}

\def\pr{{\partial}}
\def\al{\alpha}

\def\rh{{\rho}}

\def\MM{{\mathcal M}}
\def\NN{{\mathcal N}}
\def\II{{\mathcal I}}

\def\FF{{\mathcal F}}

\def\HH{{\mathcal H}}

\def\SS{{\mathcal S}}
\def\NN{{\mathcal N}}

\def\KK{{\mathcal K}}
\def\Lie{{\mathcal L}}
\def\Lieh{\widehat{\Lie}}

\def\RR{{\mathcal R}}
\def\QQ{{\mathcal Q}}

\def\Lie{{\mathcal L}}

\def\B{{\bf B}}
\def\D{{\bf D}}
\def\F{{\bf F}}
\def\H{{\bf H}}

\def\J{{\bf J}}
\def\M{{\bf M}}

\def\O{{\bf O}}

\def\R{{\bf R}}
\def\P{{\bf P}}
\def\U{{\bf U}}

\def\W{{\bf W}}
\def\Z{{\bf Z}}
\def\K{{\bf K}}
\def\T{{\bf T}}
\def\E{{\bf E}}

\def\g{{\bf g}}

\def\piX{\,^{(X)}\pi}

\def\GaX{\,^{(X)}\Ga}
\def\piZ{\,^{(Z)}\pi}
\def\PZ{\,^{(Z)}P}

\def\RRR{{\mathbb R}}

\def\f12{{\frac 1 2}}

\def\dual{{\,^*\,}}

\def\th{\theta}

\def\f{\widetilde{f}}

\def\dual{\,^*}
\def\NNb{\underline{\NN}}

\def\PX{\,^{(X)} P}

\newtheorem{theorem}{Theorem}[section]
\newtheorem{lemma}[theorem]{Lemma}
\newtheorem{proposition}[theorem]{Proposition}
\newtheorem{corollary}[theorem]{Corollary}
\newtheorem{definition}[theorem]{Definition}
\newtheorem{remark}[theorem]{Remark}

\numberwithin{equation}{section}

\begin{document}
\title{Rigidity Results   in General Relativity:  a    Review}
\author{Alexandru D. Ionescu}
\address{Princeton University}
\email{aionescu@math.princeton.edu}
\author{Sergiu Klainerman}
\address{Princeton University}
\email{seri@math.princeton.edu}


\begin{abstract}
Despite  a    common perception in the physics  community,    the Black Hole Rigidity   problem   remains     wide open   when   one removes the 
 highly restrictive real analyticity assumption  underlying the classical  results.
 In  this survey   we  review  the  progress made   in the last ten years   in      understanding the conjecture  in      the  more realistic 
 setting of smooth  spacetimes.     We  review 
  both local and global results and discuss the     new mathematical ideas  behind them.   We present  three types of  global  results which assert,
  under   somewhat different assumptions, that any  stationary solution      closed   to a non-extremal Kerr   must be  isometric   to  a  a non-extremal Kerr,    whose 
         parameters $a, M$ are determined  by  their    ADM mass   and angular momentum.  The results illustrates  an important  geometric 
          obstruction          in   understanding the  full rigidity  problem,        the  possible  presence  of                    trapped null geodesics perpendicular to  
            the stationary Killing  vectorfield.  The key insight in all these results  is that such   null geodesics are    non-existent  in  any  non-extremal Kerr 
            and thus, roughly, in any small  perturbation  of  it.

\end{abstract}
\maketitle
\tableofcontents
\section{Introduction} 
A fundamental conjecture in  General Relativity\footnote{See reviews 
by  B. Carter \cite{Ca2} and   P. Chusciel   \cite{Chrusc-Rev}, \cite{Chrusc-Rev2}, 
 for a history  and  review of the current status of the conjecture.}
asserts that the domains of outer communication
of regular\footnote{The notion of regularity 
needed here  requires a careful discussions
 concerning the geometric hypothesis 
 on the space-time. }, stationary, four dimensional,  vacuum black hole solutions
are isometrically diffeomorphic to those of 
Kerr black holes.  One expects,  due to gravitational radiation,
 that general, asymptotically flat, dynamic,  solutions of the Einstein-vacuum
  equations settle down, asymptotically,  into a stationary regime. A similar scenario is expected   to hold true in the presence of matter. Thus
  the conjecture, if true, would characterize all possible asymptotic states 
  of the general evolution.

 So far the conjecture has  been
resolved, by combining results of Hawking \cite{HE},
 Carter \cite{Ca1}, and Robinson \cite{Rob},   under the additional hypothesis of  non-degenerate
horizons and \textit{real analyticity}  of the space-time.  
The assumption of real analyticity, however,   is both hard to justify 
and difficult to dispense of.   
 One  can show, using standard elliptic theory, 
that stationary  solutions are real analytic in regions 
where the corresponding  Killing vector-field  $\T$  is time-like, 
but there is no reason to expect the same result to hold true 
in the ergo-region (in a  Kerr spacetime $\KK(a, m),  0<a<m$
  the Killing  vector-field $\T$, which is time-like in 
the asymptotic region, becomes space-like in the ergo-region).
In view of the  relevance    of the conjectured result to the
 general problem of evolution,  i.e.    the final state conjecture, 
  there is also no reason to expect that, by losing gravitational radiation,
 general  solutions become, somehow, analytic. Thus the assumption
  of analyticity is a   fundamental   limitation of the  present uniqueness results\footnote{  The results  based on analyticity
   can be reformulated as     a  proof of the fact that there can be no   other explicit   stationary  solutions.   Note also that the case of static solutions  
      has been treated in full generality, without assuming analyticity,  by  Israel \cite{I}  and   Bunting-Masood ul Alam.   \cite{Bu-M}.}.
    Here is a more precise  version of the   Carter-Robinson-Hawking result.  
   
\begin{theorem}[Carter-Robinson-Hawking]
The domain of outer communications of a \textit{real analytic}  regular, stationary (i.e. there exists a killing vectorfield $\T$ which is timelike in the asymptotic region),  four dimensional,  vacuum black hole solution is isometrically diffeomorphic to the domain of outer communications of a Kerr black hole. 
\end{theorem} 
The theorem  relies       on the following  steps.
\begin{enumerate}
\item     Bases  on the observation that,    \textit{though a general
  stationary space  may seem  quite complicated, its behavior along the
   event horizon is remarkably  simple,}  Hawking has shown that
   in addition to the original, stationary, Killing field, which has to be tangent 
   to the event horizon, there must exist, infinitesimally along the horizon and tangent to its generators,  an additional Killing  vector-field.  
   \item   In the case of a non-degenerate    horizon       the Hawking         vectorfield  can be \textit{extended}, 
    by standard hyperbolic  PDE  techniques,   to  the   full   domain of dependence  of  the     horizon, see \cite{FrRaWa}.  
     \item    The   extension  problem is   however \textit{ill posed}   in the complement   of the domain of dependence, i.e. in the domain of outer communication  of the black hole.  To overcome this difficulty   Hawking assumes analyticity     and extends     the vectorfield by
  a    Cauchy-Kowalewski type argument.     In this step           the  field  equations   are no longer used; 
   the assumption  of analyticity, which in effect replaces  the Einstein equations by the Cauchy-Riemann  equations,    completely trivializes the problem.  
       \item   As a consequence of the previous step,   the   space-time under consideration is not just stationary
    but also axi-symmetric, situation     for which  Carter-Robinson's  uniqueness theorem \cite{Ca1}, \cite{Rob} applies.  It is interesting to remark that  this   final  step  does not require analyticity.   
\end{enumerate}
A    similar result holds    for the Einstein-Maxwell equations.    Namely   the only   real analytic, stationary,  regular asymptotically flat    solutions    of the Einstein-Maxwell equations   belong to the Kerr-Newman family. The reduction to  the axially symmetric case, due to        Hawking, follows precisely the same argument  as in the vacuum case. The  rigidity of  stationary, axially symmetric    solutions  is due to 
Mazur  \cite{Mazur}, see also \cite{Chrusc-Rev2}.

  The goal of this   article  is to       review       recent  results    which aim to    prove the conjecture 
   without   appealing to analyticity.   We focus our  discussion  to the   case of the  vacuum, but we will also mention
   some of the more interesting    extensions to the   case of the Einstein-Maxwell  equations.

      We start  with a    discussion, in section 2,     of      local extension results 
    for Killing vectorfields.   The setting is very general; we consider  a   Killing vectorfield   $\Z$   defined in a   domain
    O of a   Ricci flat, smooth,   pseudo-riemannian   manifold   $(\M, \g)$   and consider the question    of whether  
      $\Z$  admits  a   smooth  Killing extension 
  in  a full neighborhood   of a point $p$ in the boundary $\pr O$.   It turns out    that the  answer is affirmative if  
  the boundary verifies    what we  call  the \textit{strict null convexity  condition}. This condition,  concerning   the behavior of
  null   geodesics   tangent to $\pr O$ at $p$,   is  automatically satisfied       on  a  Riemannian manifold, but imposes 
   a serious restriction    if $\g$ is  Lorentzian.  If in addition  the manifold    
   admits a nowhere vanishing\footnote{It suffices  to      consider a  Killing vectorfield   defined     in a neighborhood of the point
    $p\in\pr O$.}   Killing vectorfield  $\T$,  which commutes  with $\Z$ in $O$,   we  show that  $\Z$ can be extended  past $p$
    under  a weaker assumption which we call $\T$-strict  null convexity.        This is a condition which  affects only 
    the null geodesics   at $T_p(\pr O)$  which are orthogonal to  $\T$. 
    
     It is important to stress  here  that    in the particular case      of   Kerr  space-time  $\KK(a, m)$,  
      domains of the form $r_* <   r < R$ with $r$     
        the   usual Boyer-Lindquist  coordinate and $r_*$ its value on the horizon   
         are  not, in general, strictly null-convex     at    $r=R$   but are   all  strictly $\T$-null convex,  where $\T$ is
           the  stationary Killing  field    of the Kerr solution. This   fact, first discovered in \cite{Ion-K2},\cite{Ion-K1},    plays a fundamental role      in the    global results 
           discussed  in section 3  of this  paper.  
           
            The null convexity condition   is a particular      instance of the more general  
    pseudoconvexity   condition\footnote{Which applies    to general, scalar linear partial differential operators} 
       of Calderon-H\"ormander.             It is  a necessary condition  to derive uniqueness  
        results     for ill posed problems\footnote{Problems  where  existence is   by no means guaranteed.}
         based on Carleman  type estimates.  It is not a priori clear  that   the same  necessary condition  is relevant to  our extension problem.      The main   goal of section 2  is to describe the   geometric ideas by which  
            the extension problem can  in fact   be turned into an  unique continuation problem.     The results  are 
              stated in theorems \ref{extthm0}, \ref{extthm1}.      Though they  are both very  general (they hold for arbitrary semi-Riemannian 
              manifolds!) they     rely in an essential way of   the Ricci flat condition.    We also review      related   local extension results,
               see  theorems  \ref{Hawking1}, \ref{Hawking2},               for the  Hawking  vectorfield                 in a neighborhood  of   a bifurcate horizon.

            In section 3  we     discuss     three   global     results, see   \cite{Ion-K2},    \cite{AlIoKl2}, \cite{AlIoKl3},  concerning  the black hole uniqueness  problem, 
            which assert,  under   somewhat different assumptions, that any  stationary solution      closed   to a non-extremal Kerr   must be  isometric   to  a  a non-extremal Kerr,    whose          parameters $a, M$ are determined  by  their    ADM mass   and angular momentum. 
                    They are all  based           on    specific    regularity, non-degeneracy  and asymptotic  flatness    assumptions  
                   discussed in subsection \ref{sect:stationary}.

              The  first two         results       are   based on the local  characterization     of the    Kerr  solution, due to Mars  \cite{Ma1},   by the vanishing
               of the  so called   \text{Mars-Simon tensor}  $\SS$.      In  theorem \ref{Thm:Gl1}   we   make an assumption   on the bifurcation
                sphere of the          horizon   which  implies  that $\SS$ vanishes   along the horizon.    We then derive a    wave equation  for
                $\SS$  and      show, by unique continuation results,   that $\SS$  must vanish everywhere.   In  theorem  \ref{Thm:Gl1} 
                 we assume instead that $\SS$   is     sufficiently small  and rely  on the extension results    discussed in section 2  to show
                  that the spacetime is  axially symmetric. The rigidity  result then follows   by    applying    the Carter-Robinson theorem. 
                  Both results assume the presence of a unique non-degenerate horizon.   This condition was  later  removed     by Wong and Yu in   \cite{Wong-Yu} by an ingenious argument based on the mountain pass lemma.
               
               The third rigidity   result differs substantially from the other two in that we only make a smallness  assumption on  
                the bifurcate sphere. More precisely we assume that    the stationary vectorfield is small on the bifurcate  sphere
                 and deduce  that the entire  domain of outer communication is  isometric  to that     of a Kerr solution with small angular momentum.
                 This   is  first  uniqueness    result, in the  framework of smooth, 
asymptotically flat,   stationary solutions,   which combines local considerations near the horizon, via Carleman estimates, 
with information obtained by    global  elliptic   estimates.  

    These  results illustrates  an important  geometric 
          obstruction          in   understanding the  full rigidity  problem,        the  possible  presence  of                    trapped null geodesics perpendicular to  
            the stationary Killing  vectorfield.  The key insight\footnote{ A related  fact     plays   a fundamental role  in      recent   linear stability  results     
             concerning solutions  of the scalar  wave  equation in  a non-extremal  Kerr, see   \cite{DRS}   and the references therein. } in all these results  is that such   null geodesics are    non-existent  in  any  non-extremal Kerr 
            and thus, roughly, in any small  perturbation  of  it.

In the last section we  formulate,   together with  S. Alexakis,   a   general conjecture  which illustrates    the importance of   trapped  null geodesics  
 perpendicular to $\T$ and thus the importance of developing  strategies    based on global considerations, not just    on unique continuation methods starting 
  from the horizon.    
   
    \noindent    {\bf Acknowledgement.} We would like to thank S. Alexakis for reading the    paper and  making very useful suggestions.

          \section{Local Rigidity Results} In this  section  we revisit the  extension problem for Killing vector-fields
in  smooth Ricci flat Lorentzian manifolds and its  relevance to the  black hole rigidity
 problem. In the most general
situation the problem can be stated as follows:

\textit{Assume $(\M,\g)$ is a given smooth pseudo-riemannian manifold, $O\subseteq\M$  is an open subset, and $Z$ is a  smooth Killing  vector-field in $O$.  Under what assumptions does $Z$ extend (uniquely) as a Killing vector-field in $\M$?}

A classical result\footnote{See \cite{No}. We rely here on the version of  the theorem given in \cite{Chrusc2} .} of Nomizu establishes such a unique  extension   provided that the metric is real analytic, $\M$ and $O$ are connected and $\M$ is simply connected.  The result has been used, see \cite{HE} and \cite{Chrusc2},    to reduce  the 
black hole rigidity problem, for  real analytic stationary solutions of the Einstein field equations, to the 
 simpler case  of axial symmetry  treated  by the Carter-Robinson theorem.  This reduction  
 has been    often   regarded as  decisive,  especially in the physics literature,    without  a clear understanding of the sweeping  simplification   power  of the analyticity assumption. 
  Indeed the remarkable thing about Nomizu's
 theorem,  to start with,  is the fact  the metric is not assumed to satisfy any specific equation. Moreover no assumptions are needed  about  the boundary of $O$
 in $\M$ and the result is  global with only minimal assumptions on the topology
 of $\M$ and $O$.
The result  is   clearly wrong in the case of smooth manifolds $(\M, \g)$ which are not real analytic. To be able to say anything meaningful we need to  both  restrict the metric $\g$   by realistic  equations and   make specific assumptions about the boundary of $O$.   Local and global assumptions   are also need to be carefully separated.

\medskip
 In this  section  we limit our  attention  to a purely local  description of the extension problem in the 
 smooth case. 
We  assume that $(\M, \g)$ is a non-degenerate Ricci flat, pseudo-riemannian metric  i.e.
\begin{equation}\label{Ricci}
\mathrm{\bf Ric}(\g)=0.
\end{equation}

We define  the following crucial concept\footnote{In our previous papers we  have   
   used  the broader   terminology   of   \textit{pseudo-convexity}  condition, which applies  to a  given scalar  linear PDE.   }.

 \begin{definition}\label{def:null-convex}
 A domain $O\subset \M$ is said to be   strictly     null-convex at a boundary point  $p\in \pr O$
 if there exists a    small  neighborhood $U$  of $p$   and  a  smooth    (defining)     function
   $h :U\to\RRR$   such that $O\cap U=\{x\in U: h(x)<0\}$, non degenerate at $p$   (i.e. $dh(p)\neq 0$)   verifying  the  following  
    null-convexity  condition at  $p$, for all null vectors  $X\in T_p(\M)$  tangent to $\pr O$  ( i.e. $X(h)=0$),
 
\begin{equation}\label{qual1.1}
\D^2h (X,X)(p)<0
\end{equation}

\end{definition}
It is easy to see that   \eqref{qual1.1}, does not depend on the choice of the defining function $h$. The strict  null-convexity  condition is automatically satisfied  if the metric $\g$
  is Riemannian. It is also satisfied   for  Lorentzian  metrics $\g$
    if  $\pr O$ is space-like at $p$, but it imposes serious restrictions for time-like  hypersurfaces.    It  clearly  fails
 if $\pr O$ is null in a neighborhood of $p$.
  Indeed in that case  we can choose  the defining function
 $h$ to be optical,  i.e.,
 \begin{equation}\label{fail}
\D^\a h \D_\a  h =0\quad
\end{equation}
 at all  points  of $\pr O$  in a neighborhood of  $p$, and thus, choosing $X^\a=\D^\a  h $, we  have,  $$X^\a X^\b  \D_\a \D_\b h =\frac 1 2 X (\D^\a  h \D_\a h)=0.$$
 One can also  show that  unique  continuation fails    in this case.

Under the assumption  that  $\M$ contains a   Killing vectorfield  $\T$  we also define  the following variant of the null convexity condition.
 \begin{definition}\label{def:null-convex-T}
 The  domain $O\subset \M$ is said to be   strict  $\T$- null-convex at a boundary point  $p\in \pr O$
 if the  defying function $h $  at $p$ is $\T$ invariant  and    verifies   the   convexity  condition \eqref{qual1.1}
  for all null vectors  $X\in T_p(O)$   which are  orthogonal to $\T$.

\end{definition}
The following general   extension principle  was proved in \cite{Ion-K3}. A previous, related,       version  appeared in \cite{AlIoKl}.

 \begin{theorem}\label{extthm0}
Assume that $(\M, \g) $ is a smooth $d$-dimensional  Ricci flat,  pseudo-riemannian manifold  and  $O\subseteq \M$ is a strongly null-convex domain at a point $p\in\pr O$. We assume that $Z$   is a  Killing vectorfield  in $O$. Then $Z$ extends as a Killing to a neighborhood of the point $p$ in $\M$.    
 \end{theorem}
 Using  similar  techniques one can also prove the following.
\begin{theorem}
\label{extthm1}
If $(\M, \g) $ admits a (nowhere vanishing)  Killing vectorfield   $\T$and   $Z$    is a Killing vectorfield   in $\O$  which commutes   with $\T$,    then the  same extension result  holds true
 if we replace    strict   null convexity by the weaker 
   strict $\T$-null convexity condition. Moreover the extended $Z$ continues  to   commute with $\T$.

\end{theorem}
The proof of both theorems  is  based on the following ideas\footnote{  In \cite{AlIoKl}     similar results were proved    using  a frame   dependent approach.  }.
\begin{enumerate}
\item  Extend the vectorfield  $Z$  in a full  neighborhood of the point $p$ by solving a Jacobi type equation
along    a family of  congruent  geodesics     transversal to $\pr O$. In the case of theorem \eqref{extthm0}
one needs   to make sure    that  extended  $Z$  still commutes  with $\T$.  This  can easily be done   by  
 choosing a congruence     left invariant  by $t$, i.e. such that  the     generator $L$ of  the   congruence    commutes with $\T$. t
\item  Derive   a closed  system of      covariant wave equations for   a modified version   of  the Lie  derivative
of  the curvature tensor $R$, denoted $W$,    coupled with transport equations  for
   the  deformation tensor  $\piZ$  of the extended $Z$.
\item Use  a  unique continuation argument   to show  that both     $W$ and     $\piZ$  have to vanish in a full neighborhood
of $p$.  To implement the continuation criterion  one needs      the  strict   null convexity  conditions     in  the definitions 
\ref{def:null-convex} and \ref{def:null-convex-T}.
\end{enumerate}

We start with a few  general results:
   \begin{lemma}\label{le:commute}
For  arbitrary $k$-covariant  tensor-field  $V$ and vector-field  $X$ we have,
\begin{equation}\label{commute}
\D_\b(\Lie _X V_{\a_1\ldots\a_k})-\Lie_X(\D_\b V_{\a_1\ldots\a_k})=
\sum_{j=1}^k\GaX_{\a_j\b \rho}V_{\a_1\ldots} \,^\rho\,_ {\ldots  \a_k}
\end{equation}
where $\piX=\Lie_X\g$  is the deformation tensor of $X$ and,
\begin{equation*}
\GaX_{\a\b\mu}:=\frac 1 2 (\D_\a\piX_{\b\mu}+\D_\b\piX_{\a\mu}-\D_\mu\piX_{\a\b}).
\end{equation*}
\end{lemma}
\begin{lemma}  \label{le: nice}  
Let $X$ be a vectorfield with deformation tensor $\piX$ and    define,
\beaa
\PX_{\a\b\mu}&:=&(1/2)(\D_\a\piX_{\b\mu}-\D_\b\piX_{\a\mu}).
\eeaa
Then,

\begin{equation}\label{nice1}
 \D_\nu \PX_{\a\b\mu}-\D_\mu\PX_{\a\b\nu}=(\Lie_X\R)_{\a\b\mu\nu}-(1/2)\piX_\a\,^\rho\R_{\rho\b\mu\nu}-(1/2)\piX_\b\,^\rho\R_{\a\rho\mu\nu}
\end{equation}
where $\R$ is  the Riemann curvature tensor of the metric $\g$.
\end{lemma}
\noindent Note that,
\bea
\GaX_{\a\mu \b }&=&\PX_{\a\b\mu}+ \frac 1 2 \D_\mu  \piX_{\a\b}
\eea

Recall that a Weyl field on $\M$ is a a  four  covariant   tensor, trace-less  tensor,   verifying all the symmetries of the Riemann curvature tensor.
Note that the  Lie derivative  of a Weyl    field may fail to   to  have vanishing  trace.  The leads us  to the  following modified  definition.  
\begin{definition}
Given a Weyl field $W$,      $X$  an arbitrary  vectorfield and $\om_{\a\b}$  an arbitrary $2$-form on $\M$, we define,
\beaa
\Lieh_{X, \om}  W:= \Lie_X W-\frac 12(\pi+\om) \odot  W
\eeaa
where, for any  $2$-tensor  $B$,  $B\odot W$ denotes the tensor, 
 \beaa
 ( B\odot W)_{\a\b\ga\de}:=B_\a\,^\la W_{\la\b\ga\de}+B_\b\,^\la W_{\a\la\ga\de}+
 B_\ga\,^\la W_{\a\b\la\de}+B_\de\,^\la W_{\a\b\ga\la}.
  \eeaa
\end{definition}
   \begin{lemma}
   The tensor $\Lieh_{X, \om}  W$  defined above is a Weyl field.
   \end{lemma}

 \subsection{Proof of Theorem  \eqref{extthm0}}
 To prove the theorems we   first extend $Z$    past $p$ according to the  following  equation
\begin{equation}
\D_L \D_L Z=\R(L, Z) L\label{eq:extZ},
\end{equation}
Let $\piZ=\Lie_Z \g$ be the deformation tensor of $Z$.    To show that $\piZ \equiv 0$  in a neighborhood  $U$  of $p$ 
we need to prove that $\piZ,  \PZ$,  $\Lie_Z\R$   all vanishes,  simultaneously     in $ U$.  The idea is  to    try
 to derive transport equations  for $\piZ$  and $\PZ$,  along the   geodesics generated by $L$,  coupled  to   
  a    covariant   wave equation  for      $\Lie_Z\R$. To  do this we 
 will need  however to    redefine slightly  the main quantities.  The crucial    ingredient 
 which makes  possible    to derive useful  transport  equations is the following.
 \begin{lemma}
 If  $Z$ is extended according to \eqref{eq:extZ}      then  the deformation tensor $\pi:=\piZ$ of $Z$  verifies
 \bea
  \pi_{\a\b} L^\b=0.
 \eea
 \end{lemma}
 \subsubsection{The main  coupled system}
To derive the desired transport equations  we   would also need    that $P=\PZ$   verifies   $P_{\a\b\mu} L^\mu=0$. 
This is not true however and we are forced to introduce  the modification,
\bea
\P_{\a\b\mu}&=  & P_{\a\b\mu}-\frac 1 2    \D_\mu \om_{\a\b}    =    (1/2)(\D_\a\pi_{\b\mu}-\D_\b\pi_{\a\mu}-\D_\mu \om_{\a\b})               
\eea
with $\om$ a  $2$-form  chosen  precisely  such that  $\P_{\a\b\mu} L^\mu=0$.
This leads to the following.
 
 \begin{lemma}
 If we define $\omega$ in $\M$ as the  solution of the transport equation
  \begin{equation}\label{maincor}
\D_L\omega_{\a\b}=\pi_{\a\rho}\D_\b L^\rho-\pi_{\b\rho}\D_\a L^\rho,
\end{equation}
with $\omega=0$  in $O$, then
\begin{equation}\label{LP}
L^\mu \P_{\a\b\mu}=0, \qquad L^\b\om_{\a\b}=0,  \qquad L^\mu \B_{\a\mu}= 0 = L^\mu\B_{\mu\a}  \text{ in \,\, } \M.
\end{equation}
where   $\B=\pi+\om$, 
  \end{lemma}
 With these preliminaries   one  can   easily  derive  transport equations for the tensors $\B$ and $\P$  along the  geodesics generated by $L$.
\begin{proposition}\label{pro5} Let $\B, \P $ as above and   $\W:=   \Lieh_{Z,\om}\R$.  We have,
\bea
\D_L \B_{\a\b}&=&L^\rho \P_{\rho\b\a}-\D_\al L^\rho \B_{\rho\b},\label{yi2}\\
\D_L \P_{\a\b\mu}&=&L^\nu \W_{\a\b \mu\nu} +L^\nu {\B_{\mu}}^\rho\R_{\a\b\rho\nu}-\D_\mu L^\rho \P_{\al\be\rho} \label{yi3}
\eea

  \end{proposition}
  \begin{definition}
  \label{def:1}
 By convention, we let $\mathcal{M}({}^{(1)}B,\ldots,{}^{(k)}B)$ denote any smooth ``multiple'' of the tensors ${}^{(1)}B,\ldots,{}^{(k)}B$, i.e. any tensor of the form
\begin{equation}\label{mnotation}
\mathcal{M}({}^{(1)}B,\ldots,{}^{(k)}B)_{\al_1\ldots\al_r}={}^{(1)}B_{\be_1\ldots\be_{m_1}}{}^{(1)}{C_{\al_1\ldots\al_r}}^{\be_1\ldots\be_{m_1}}+\ldots +{}^{(k)}B_{\be_1\ldots\be_{m_k}}{}^{(k)}{C_{\al_1\ldots\al_r}}^{\be_1\ldots\be_{m_k}},
\end{equation}
for some smooth tensors ${}^{(1)}C,\ldots, {}^{(k)}C$ in $\M$.
 \end{definition} 
 With this definition  proposition \ref{pro5} takes the form,
 \beaa
\D_L \B&=&\MM(\W, \B, \P), \qquad   \D_L \P=\MM(\W, \B, \P),   \\
 \eeaa
To get a closed  system  it  remains to establish an equation for $  \W= \Lieh_{Z,\om}\R$.   This is achieved  by the following.
\begin{lemma}
Let $\W:= \Lie_{Z, \om}\R$,   with $\om$ an arbitrary $2$ form. Then, with the definitions made above,
\bea
\label{Bianchi:W}
 \D^a\W_{\a\b\ga\de}&=&\J_{\b\ga\de}
 \eea
 where,
 \beaa
 \J_{\b\ga\de}&=& \B^{\mu\nu}\D_\nu \R_{\mu\b\ga\de}+\g^{\mu\nu}\P_{\mu\rho\nu}\R^\rho\,_{\b\ga\de}+\P_{\b\nu\mu}\R^{\mu\nu}\,_{\ga\de}+\P_{\ga\nu\mu}   \R^{\mu}\,_\b\,^\nu\,_{\de}
  +\P_{\de\nu\mu}\R^{\mu}\,_{\b\ga}\,^\nu.\nn
  \eeaa
\end{lemma}
\begin{proof}
Follows easily from       the definition of $W$ and  lemma  \ref{le:commute}  applied  to  the curvature  tensor $\R$
 and vectorfield $Z$.
 
\end{proof}
Differentiating  \eqref{Bianchi:W} and using  the symmetries of $\W$  we easily deduce  (see  7.1.   in \cite{}),
\beaa
\D^\rho\D_\rho \W_{\al\be\mu\nu}=\mathcal{M}(\B,\D \B,\P\,\D \P,\W)_{\al\be\mu\nu}.
\eeaa

We  have thus derive   the closed system,

\beaa
\D_L\B&=&\MM(\W, \B, \P)\\
\D_L \P&=&\MM(\W, \B, \P),\\
\square  \W&=&\MM(\W, \B, \D\B,   \P, \D\P)
\eeaa
with  the notation $\MM(W, B, P)$ explained below.

\subsubsection{Unique continuation argument}
Once we have our coupled system it remains to   prove   the simultaneous   vanishing of $\B, \P, \W$  by a unique continuation argument.
 More generally  we consider solutions of systems of equations of the form,
\begin{equation}
\label{gen-uniqueness}
\begin{cases}
&\square_\g S=\mathcal{M}({}^{(1)}B,\ldots,{}^{(k)}B,S,\D S)\\
&\D_L{}^{(i)}B=\mathcal{M}({}^{(1)}B,\ldots,{}^{(k)}B,S,\D S),\qquad i=1,\ldots,k.
\end{cases}
\end{equation}
Theorem \ref{extthm0} is now an immediate  consequence  
 of the following.
 \begin{proposition}
 \label{prop:uniq1}
 Let $(\M, \g)$ be a general pseudo-riemannian manifold, $O$ a domain in $\M$ verifying  the strict  null-convexity condition   at $p\in \pr O$. 
 Assume    given a  collection of    tensorfields   $S, B$  on $\M$, and a vectorfield $L$   verifying  \eqref{gen-uniqueness} in a neighborhood of $p$.
 Then, if $(S, B)$ vanish  in a neighborhood  of  $p$, in $O$,  they also vanish in a full neighborhood of $p$.    
 \end{proposition} 
 
   Expressing 
 the equations \eqref{gen-uniqueness}  in local coordinates\footnote{Using  definition \ref{def:1}  and also  enlarging $M$ and restricting  the neighborhood   $U$ of $p$ as necessary.},    we  can easily  reduce  the statement of proposition 
 \ref{prop:uniq1} to the following statement.
 \begin{proposition}
 \label{prop:unique-cont}
 Assume that  $h$ is a   strictly  null convex     defining function for $O$  in a neighborhood  $U$ of $p\in\pr O$.          
 Assume   given smooth   function  $G_i, H_j$ 
 $i=1,\ldots,I$, $j=1,\ldots,J$,  which  satisfy the following  differential inequalities in    $U$ of  $p\in \pr O$,
\begin{equation}\label{ext68}
\begin{cases}
&|\square_\g G_i|\leq M\sum_{l=1}^I(|G_l|+|\partial G_l|)+M\sum_{m=1}^J|H_m|;\\
&|L(H_j)|\leq M\sum_{l=1}^I(|G_l|+|\partial G_l|)+M\sum_{m=1}^J|H_m|,
\end{cases}
\end{equation}
for any $i=1,\ldots,I$, $j=1,\ldots,J$. Then,  if     the function $G, H$  vanish in    $U\cap O$     then they also vanish  
in a full,  small,  neighborhood of $p$. 

 \end{proposition}

The  proof of the proposition  is based on Carleman estimates. The first step  is to  obtain a quantitative version of    our null-convexity condition.
\begin{lemma}
Assume the defying  function $h$  is  strictly    null -convex at $p$.  There exists a constant  $M>0$, depending only on bounds for the metric  $g$ and its derivatives\footnote{with respect to a fixed   system of coordinates   at $p$}    in a fixed  coordinate   neighborhood  $V$  of $p$, and $\mu\in[-M, M]$
as well as   a  small  neighborhood   $U\subset  V$    of $p$  such that, for any vectorfield $Y=Y^\a\partial_\a$   
\bea
\label{null-convex-strong}
\begin{cases}
&\quad \qquad \qquad \qquad \qquad \qquad \qquad| d  h|  \ge  M^{-1}\\
&Y^\a Y ^\b  \left( \mu    g_{\a \b}- D_\a   D_\b h\right)+M|Y (h )|^2\,\, \geq  M^{-1}|Y|^2,
\end{cases}
\eea
uniformly, at all points of $U$, with $|Y|^2= (Y ^0)^2+(Y^1)^2+\ldots+(Y^d)^2$..

\end{lemma} 

 Here is also   a quantitative version of    the    $\T$-null-convexity condition. 
 
\begin{lemma} 
Assume the defying  function $h$  is  strictly  $\T$-null-convex at $p$
There exists a constant  $M>0$, depending only on bounds for the metric  $g$ and its derivatives\footnote{with respect to a fixed   system of coordinates   at $p$}    in a fixed  coordinate   neighborhood $V$   of $p$,  a   constant $\mu\in[-M, M]$ and 
  a sufficiently small neighborhood   $U\subset  V$    of $p$  such that, for any vectorfield $Y=Y^\a\partial_\a$   
\bea
\label{null-convex-strong-T}
\begin{cases}
&\quad \qquad \qquad \qquad \qquad \qquad \qquad  \qquad \qquad \qquad |d h|  \ge  M^{-1}\\
&Y^\al Y  ^\be( \mu   g_{\al\be}- D_\al   D_\be h )+M\left( |Y (h )|^2  +| g(\T , Y) |^2 \right)     \geq M^{-1}|Y|^2,
\end{cases}
\eea
uniformly, at all points of $U$, with $|Y|^2= (Y^0)^2+(Y^1)^2+\ldots+(Y^d)^2$.. 
\end{lemma}

 The   proof of proposition  \eqref{prop:unique-cont}  can be reduced  to two Carleman estimates. 
  The first, and   by far   the more important one,  concerns  the scalar wave operator $\square_\g$.
  To state it we assume    that  the defining function $h$           of the domain $O$, near $p\in \pr O$,  verifies  
 \eqref{null-convex-strong}  in a full neighborhood $ U_1$  of 
 a point $p$ with $h(p)=0$.  Let  $U_\ep$ be small  neighborhoods of 
 $p$ such that $|h|\le  2^{-1}\ep$    in $U_\ep$ and define  the weight functions,
  $f_\ep: U_\ep \longrightarrow \RRR$
 \beaa
 f_\ep :=\log (\ep+h+e_p)
 \eeaa
 where $e_p$ is a small  perturbation of  $\ep+h$.
  such  that  the    weights  $f_\ep$   verify \eqref{null-convex-strong} in $U_\ep$, uniformly in $\ep>0$

We are now ready to state      our  main   Carleman estimate.       

\begin{proposition}
\label{prop:Carleman1}
 If $f_\ep$ are as above,     there exists  a sufficiently  small     $\ep>0$  and a  large constant $C_\ep>0$ 
   such that, for all $\phi \in C_0^2(U_\ep)$ and all sufficiently large  $\la>0$,
     \begin{equation}\label{ca1}
\lambda \cdot \|e^{-\lambda f_\ep }\cdot \phi\|_{L^2}+
\|e^{-\lambda f_\ep }\cdot D \phi\|_{L^2}\leq C_\ep    \lambda^{-1/2}\cdot \|e^{-\lambda f_\ep}\cdot\square_\g \phi\|_{L^2},
\end{equation}

\end{proposition}

\begin{remark}
A more general version of the  Carleman estimate  \eqref{ca1},  adapted to   the notion of $\T$-null convexity 
is given in \cite[section 3.2] {Ion-K2}
\end{remark}

 We also need a Carleman estimate to  deal with  the     ODE  part of  our  system. This is considerable easier,   no additional restrictions are needed, see \cite[Lemma 3.4]{AlIoKl}.

\subsection{Existence   results in a neighborhood of the horizon}
The methods  discussed in the previous   subsections  can be applied to construct
Killing vectorfields in a neighborhood of    a bifurcate horizon  of stationary     vacuum 
solutions. One can in fact present the result   without  reference to  
 stationarity as follows.   
 
 Let  $(\M,\g)$ to be a smooth\footnote{$\M$ is assumed to be a connected, oriented, $C^\infty$ $4$-dimensional manifold without boundary.} vacuum Einstein space-time. Let $S$  be an embedded spacelike $2$-sphere in $\M$ and let $\NN, \NNb$ be the  null boundaries of the causal set of $S$, i.e. the union of the causal  future and past of $S$. We fix $U$ to be a small neighborhood of $S$ such that both $\NN,\underline{\NN}$ are regular, achronal, null hypersurfaces in $U$ spanned by null geodesic generators orthogonal to $S$. We say that  the triplet $(S, \NN,\underline{\NN})$ forms a local, regular, bifurcate,  non-expanding horizon  in $U$ if both  $\NN,\underline{\NN}$ are non-expanding null hypersurfaces in $U$.  This simply means that     the      traces of the null   second fundamental forms    of $\NN$ and $\NNb$, called expansions,   
  are both vanishing  respectively on $\NN$ and $\NNb$.              Our main results are   the following:
  
  \begin{theorem}\label{Hawking1}
Given a   local, regular, bifurcate\footnote{  Hawking's original  rigidity  theorem relies instead on a non-degeneracy assumption. We note however that the two  assumptions   are    in fact related,  see \cite{Ra-Wa}.}  $(S,\, \NN,\,\NNb)$  in  
 a smooth,  vacuum Einstein space-time $(\M,\g)$,   there exists  an open neighborhood $ V\subset U$ of $S$ and a non-trivial Killing vector-field $\K$ in $\M$, which is tangent to the null generators of $\NN$ and $\NNb$. In other words, every local, regular, bifurcate, non-expanding horizon is a Killing bifurcate horizon.
\end{theorem}

\begin{theorem}
\label{Hawking2}
Under the same assumptions as  above, if in addition    there exists a Killing  vectorfield $\T$  in $ U$ 
 tangent to  $\NN\cup\NNb$ and not identically vanishing on $S$, then   there exists an  open neighborhood 
 $V\subset U$ and   a non-trivial rotational  Killing vector-field $\Z$ in $U$ which commutes with $\T$.

\end{theorem}

It was  already known, see \cite{FrRaWa}, that such a Killing vector-field exists in a small neighborhood        of $S$ intersected with the domain of dependence of $\NN\cup\underline{\NN}$, which we  could call $O$ in reference  to  theorem  \ref{extthm0}. The extension  of $\K$ to a full neighborhood of $S$ has been known to hold only under the restrictive  additional  assumption of  analyticity of the space-time (see \cite{HE}, \cite{IsMon}, \cite{FrRaWa}). The novelty of   both    theorems is the ability to construct these   local Killing  fields  in a full neighborhood of the 2-sphere $S$, without making any analyticity assumption.   Both    theorems  can be viewed as  applications of  theorems \ref{extthm0}, \ref{extthm1}
 to domains     $O$     which are    obtained by intersecting  neighborhoods of $S$ in $\M$  with the      domain of dependence   of   the bifurcate sphere  $S$.   The strict null convexity    condition is an easy consequence of  the   bifurcation (non-degeneracy)  property   of  the   boundary $(\NN\cup \NNb)\cap U$ of $O$.  Note however that $O$  is not smooth  at  points of $S$. This requires a  slight modification  of the Carleman estimates      needed in the proof of    theorems \ref{extthm0}, \ref{extthm1}. A full account of such Carleman estimates  is           given in \cite[section 3.2] {Ion-K2}.

\subsection{Counterexamples} We review a  counterexample   to Haking's rigidity theorem  in the non-analytic case. 
 Let $(\mathcal{K}(m,a),\g)$ denote the (maximally extended) Kerr space-time of mass $M$ and angular momentum $Ma$, $0\leq a<M$.  Let $\M^{(end)}$ denote an asymptotic region, $\mathbf{E}=\mathcal{I}^-(\M^{(end)})\cap \mathcal{I}^+(\M^{(end)})$ the corresponding domain of outer communications, and $\mathcal{H}^-=\delta(\mathcal{I}^+(\M^{(end)})$ the boundary (event horizon) of the corresponding white hole\footnote{A similar statement can be made on  the future event horizon $\HH^+$.}. Let $\T=d/dt$ denote the stationary (timelike in $\M^{(end)}$) Killing vector-field  of $(\mathcal{K}(m,a),\g)$, and let $Z=d/d\phi$ denote its rotational (with closed orbits) Killing vector-field. The following theorem was proved in \cite{Ion-K3}.

\begin{theorem}\label{mkathm}
Assume that $0<a<M$ and $U_0\subseteq \mathcal{K}(m,a)$ is an open set such that $U_0\cap\mathcal{H}^-\cap\overline{\mathbf{E}}\neq\emptyset.$
 Then, 
\begin{enumerate}
\item  There is an open set $U\subseteq U_0$, $U\cap \mathcal{H}^-\neq\emptyset$, and a smooth Lorentz metric $\widetilde{\g}$ in $U$
such that,
\begin{equation}\label{mka1}
{}^{\widetilde{\g}}\mathbf{Ric}=0\,\,\text{ in }U,\qquad \Lie_\T{\widetilde{\g}}=0\,\,\text{ in }U,\qquad {\widetilde{\g}}=\g\,\,\text{ in }U\setminus\mathbf{E};
\end{equation}

\item The vector-field $Z=d/d\phi$ does not extend to a Killing vector-field for $\widetilde{\g}$, commuting with $\T$, in $U$.
\end{enumerate}
\end{theorem}

In other words, one can modify the Kerr space-time smoothly, on one side of the horizon $\mathcal{H}^-$, in such a way that the resulting metric still satisfies the Einstein vacuum equations, has $\T=d/dt$ as a Killing vector-field, but does not admit an extension of the Killing vector-field $Z$. 
The crucial point here is    that the neighborhood under consideration         is away from the bifurcate sphere, where    theorems  \ref{Hawking1} - \ref{Hawking2} apply.         The     result illustrates one  of the major   difficulties  one  faces in trying to extend Hawking's rigidity result to the more realistic setting of  smooth
     stationary  solutions of the Einstein vacuum equations: unlike in the analytic situation, one cannot hope to construct an additional symmetry of stationary solutions of the Einstein-vacuum equations (as in Hawking's Rigidity Theorem) by relying only on the local information provided by the equations.
     
     The  proof   relies on  a symmetry  reduction induced by  the Killing  vectorfield   $\T$.   We  denote    the   fixed Kerr  metric  by $\g$
and      define   the  \textit{reduced}  metric
\begin{equation*}
h_{\al\be}=X\g_{\al\be}-\T_\al \T_\be,\qquad\text{ where }X=\g(\T,\T),
\end{equation*}
on a hypersurface $\Pi$ passing through the point $p$ and transversal to $\T$. The metric $h$ is nondegenerate (Lorentzian) as long as $X>0$ in $\Pi$, which explains our assumption $0<a<m$. It is well-known, see for example \cite[Section 3]{We}, that the Einstein vacuum equations together with stationarity $\mathcal{L}_\T\g=0$ are equivalent to the system of equations
\begin{equation}\label{brac}
\begin{split}
&{}^h \mbox{Ric}_{ab}=\frac{1}{2X^2}(\nabla_a X\nabla_b X+\nabla_a Y\nabla_b Y),\\
&{}^h\square (X+iY)=\frac{1}{X}h^{ab}\partial_a(X+iY)\partial_b(X+iY),
\end{split}
\end{equation}
in $\Pi$, where $X+i Y$ is the   complex      Ernst potential associated\footnote{See  section  \ref{sect:Ernst} }   to $\T$.

We then modify the metric $h$ and the functions $X$ and $Y$ in a neighborhood of the point $p$ in such a way that the identities \eqref{brac} are still satisfied. The existence of a large family of smooth triplets $(\widetilde{h},\widetilde{X},\widetilde{Y})$ satisfying \eqref{brac} and agreeing with the Kerr triplet in $\Pi\setminus\mathbf{E}$ follows by  a classic local   existence result,  solving a characteristic initial-value problem, using, for example,  the main  existence  result  in \cite{Re}. 

 One can then  we construct the     new  space-time metric $\widetilde{\g}$,
\begin{equation*}
\widetilde{\g}_{ab}=\widetilde{X}^{-1}\widetilde{h}_{ab}+\widetilde{X}\widetilde{A}_a\widetilde{A}_b,\qquad \widetilde{\g}_{a4}=\widetilde{X}\widetilde{A}_a,\qquad\widetilde{\g}_{44}=\widetilde{X},\qquad a,b=1,2,3,
\end{equation*}
 associated to the triplet $(\widetilde{h},\widetilde{X},\widetilde{Y})$, the vector-field $\T=\partial_4$, and a suitable $1$-form $\widetilde{A}$ which is defined in $\Pi$. By construction    (see  \cite[Theorem 1]{We})   this metric verifies the identities ${}^{\widetilde{\g}}\mathbf{Ric}=0$ and $\Lie_\T\widetilde{\g}=0$, in a suitable open set $U$.  Finally  one can  show that we have enough flexibility to choose initial conditions for $\widetilde{X},\widetilde{Y}$ such that the vector-field $Z$ cannot be extended as a Killing vector-field for $\widetilde{\g}$ commuting with $\T$, in the open set $U$.

\begin{remark} 
Note that the construction     of the extended metric      in theorem  \ref{mkathm}  relies  in an essential way on     the  fact  that  $\neq 0$  to  allow for a non-trivial ergo-region
 near the horizon where $\T$ is space-like.  No such result is known for $a=0$.
\end{remark}

\section{Mars-Simon tensor and global results}

\subsection{Killing vector-fields }
\label{sect:Ernst}

 In what follows we    consider $1+3$   dimensional Lorentzian manifolds
    endowed  with a Killing vectorfield  $\K$, i.e.
\bea
\D_\al\K_\be+\D_\be\K_\al=0\label{xi-killing}
\eea
We define the 2-form,
\beaa
\F_{\a\b}=\D_\al\K_\be=\frac 1 2 \big(\D_\a \K_\b-\D_\b \K_\a)
\eeaa
 as well as its Hodge dual,
 \beaa
 \dual F_{\a\b}&=&\frac 1 2 \in_{\a\b}\, ^{\mu\nu} F_{\mu\nu}
 \eeaa 
 Note that  $\dual( \dual F)=-F$.   We also define    the left and right  Hodge duals of
 the curvature tensor,
 \beaa
 \dual \R_{\a\b\ga\de} &=& \frac 1 2  \in_{\a\b}\,^{\mu\nu} \R_{\mu\nu\ga \de}, 
 \qquad  \R^*_{\a\b\ga\de} =\frac 1 2 \R_{\a\b\mu\nu}\in^{\mu\nu}\,_{\ga\de}
 \eeaa
 and note that    for a vacuum manifold, i.e.  {\bf Ric}$(\g)=0$, we have $\dual \R=\R^*$,  $\dual(\dual \R)=-\R$.
 We   also define the complex tensors, 
 \beaa
 \FF_{\a\b}= F_{\a\b}+i \dual F_{ab},\qquad 
 \RR_{\a\b\ga\de}=\R_{\a\b\ga\de}+\dual \R_{\a\b\ga\de}
 \eeaa
 Note  that  $\RR$ verifies all  the symmetries of the  curvature tensor as well  
 as\footnote{i.e. $\RR$ is a complex  valued  Weyl tensor.} $ \g^{\a\ga} \R_{\a\b\ga\de}=0$.
 Note also that both   $\FF$ and $\RR$ are self dual i.e $\dual \FF=-i \FF, \ \dual \RR=-i \RR$.

    We recall the following well known,
\begin{lemma}
\label{le:comm-K1}
For all tensor-fields  $\U$ in $\M$, if $\K$ is Killing we  have,
\bea
\D_\mu \D_\a \K_\b&=&\R_{ \la \mu\a\b}  \K^\la           \label{eq:Ricci}\\
\, [\Lie_\K, \D] \U&=&0\label{eq:com-K}
\eea
In particular, if $(\M, \g)$  has vanishing Ricci curvature then,
\beaa
\D_\mu \FF_{\a\b}&=& \R_{ \la \mu\a\b}  \K^\la  
\eeaa

\end{lemma}
\begin{corollary}
If $(\M, \g)$  is a vacuum Lorentzian manifold  endowed with a Killing vectorfield $\K$  we 
have,
\beaa
\D_{[\mu}\FF_{\a\b]}=\D_\mu \FF _{\a\b }+\D_\a \FF_{\b\mu}+\D_\b \FF_{\mu\a}=0, \qquad \D^\b \FF_{\a\b}=0.
\eeaa
\end{corollary}

We now  define the complex valued 2-form,
\begin{equation}\label{s1}
\FF_{\al\be}=\F_{\al\be}+i {\dual \F}_{\al\be}.
\end{equation}
Clearly, $\FF$ is self-dual solution of the Maxwell equations,  i.e. $ \FF\dual=(-i) \FF$ 
and
\bea
\D_{[\mu} \FF_{\al\be]}=0\label{eq:Max1'},\quad \D^\b \FF_{\a\b}=0.\label{eq:Max2'}
\eea
We define also the Ernst $1$-form associated to the Killing vector-field  $\K$,
\bea
\si_\mu&=&2(i_\K \F_\a)=2\K^\a\FF_{\a\mu}=\D_\mu(-\K^\al \K_\al)- i\in_{\mu\b\ga\de}\K^\b \D^\ga\K^\de.
\eea
\begin{proposition} The following are true,
\begin{equation}\label{Ernst1}
\begin{cases}
&\D_\mu\si_\nu-\D_\nu\si_\mu=0;\\
&\D^\mu\si_\mu=-\FF^2;\\
&\si_\mu\si^\mu=g(\K,\K) \FF^2.
\end{cases}
\end{equation}
\end{proposition}
\begin{proof}
We have,
\beaa
2^{-1}\big(\D_\mu\si_\nu-\D_\nu\si_\mu\big)&=&\K^\a\big(\D_\nu \FF_{\a\mu}-\D_\mu\FF_{\a\nu}\big)+ \D_\nu\K^\a\FF_{\a\mu}             -\D_\mu\K^\a \FF_{\a\nu}\\
&=&-\K^\a \D_\a \FF_{\mu\nu} -  \D_\mu\K^\a \FF_{\a\nu}-\D_\nu\K^\a\FF_{\mu\a} =-\Lie_\K \FF_{\mu\nu}\\
&=&0.
\eeaa
Also,
\beaa
2^{-1}\D^\mu\si_\mu&=&\K^\a \D^\mu  \FF_{\a\mu}+\D^\mu \K^\a \FF_{\a\mu}=
-F^{\a\mu}\FF_{\a\mu}=-2^{-1}\FF^2
\eeaa
The last formula    in \eqref{Ernst1}   follows easily from the lemma below.
\end{proof}
 \begin{lemma}
 \label{le:decomp}
  Introduce    the decomposition
  \beaa
  i_\K( \F)_\a=\K^\mu \F_{\mu\a},\quad    i_\K( \dual \F)_\a=\K^\mu \dual \F_{\mu\a},\quad   i_\K(\FF)_\a=\K^\mu\FF_{\mu\a}.
  \eeaa
   Clearly, $i_\K (\F), i_\K(\dual \F)$,  $i_K(\FF) $ are orthogonal to $K$
   and,
    \bea
 g(\K, \K) \F_{\a\b}=\K_\a\, i_\K(\F)_\b-\K_\b\, i_\K( \F)_\a+ \in_{\a\b\mu\nu} \K^\mu \,i_\K(\dual \F)^\nu
\eea
Also,
 \bea
 g(\K, \K) \FF_{\a\b}=\K_\a\,  i_\K(\FF)_\b-\K_\b\,  i_\K( \FF)_\a-i \in_{\a\b\mu\nu} \K^\mu\,  i_\K(\FF)^\nu
 \label{decomp-F}
\eea
In particular,
\bea
 g(\K, \K) \FF^2&=&4i_\K(\FF)^\mu  i_\K(\FF)_\mu=\si^\mu\si_\mu
 \eea
\end{lemma}
\begin{remark}
Since $d(\si_\mu dx^\mu)=0$,  if $\M$ is simply connected,  we infer that there exists a function $\si:\M \to\mathbb{C}$, called the Ernst potential, such that $ \si_\mu =\D_\mu\si$.  Note also  that    $\D_\mu\g(\K, \K)=2 \F_{\mu\la} \K^\la = - \Re\si_\mu $.   Hence 
 we can choose  the potential $\si$  such that,
 \bea
 \Re\si=-\g(\K, \K).
 \eea
 Moreover, if $(\M, \g)$ is asymptotically flat, we can choose $\si=1$   at space like infinity. 
\end{remark}
As a corollary of the lemma we also deduce,
\bea
\square \si&=& -\g(\K, \K)^{-1} \D_\mu\si \D^\mu \si \label{WM-si}
\eea
or, writing $\si=-f-i  f^* $  we deduce,
\bea
\square f&=&f^{-1}\big(\D^\mu f \D_\mu f-\D^\mu f^* \D_\mu  f ^*)\label{WM-f}\\
\square f^* &=& f^{-1} \D^\mu  f\,\D_\mu  f ^* \label{WM-f*}
\eea
In other words the pair $(x=f,  y=f^*)$ defines, whenever $f=g(K,K)\neq 0$,
a wave map  to the Poincar\'e  plane $\H:=\{(x, y)/  x>0\}$ with metric,
\beaa
ds^2=\frac{dx^2+dy^2}{x^2}
\eeaa

\subsection{Stationary Vacuum Spacetimes} 
\label{sect:stationary}    We consider  vacuum,  asymptotically flat,    $1+3$  dimensional     spacetimes    which are stationary, i.e. they possess 
 a  smooth, non degenerate, Killing  vectorfield   $\T$ which is timelike in the   asymptotic region (i.e. a neighborhood of null infinity).  More precisely we make the following assumptions:
 \begin{enumerate}
 
\item  (Asymptotic flatness.)  We assume that there is an open subset $\M^{(end)}$ of $\M$ which is diffeomorphic to $\mathbb{R}\times(\{x\in\mathbb{R}^3:|x|>R\})$ for some $R$ sufficiently large. In local coordinates $\{t,x^i\}$ defined by this diffeomorphism, we assume that, with $r=\sqrt{(x^1)^2+(x^2)^2+(x^3)^2}$,
\begin{equation}\label{As-Flat}
\g_{00}=-1+\frac{2M}{r}+O(r^{-2}),\quad \g_{ij}=\delta_{ij}+O(r^{-1}),\quad\g_{0i}=-\ep_{ijk}\frac{2S^jx^k}{r^3}+O(r^{-3}),
\end{equation}
for some   constants   $M>0, S_1, S_2, S_3$ such that   $J=[(S^1)^2+(S^2)^2+(S^3)^2]^{1/2}\in[0,M^2).$ We also assume   $\quad \T=\pr_t$ with  $t=x^0$.
We define the domain of outer communication (exterior region)
\begin{equation*}
\E=\II^{-}(\M^{(end)})\cap\II^{+}(\M^{(end)}).
\end{equation*}
  \item  (Completeness.)           We also  assume  that  $\E$ is globally hyperbolic    and      every orbit of $\T$ in $\E$ is complete and intersects  a given  spacelike   Cauchy  hypersurface   $\Sigma_0$.   We also  assume, for convenience,   that  $\Si_0$    is diffeomorphic to $\{x\in\mathbb{R}^3:|x|>1/2\}$ and agrees with the hypersurface corresponding to $t=0$  in $\M^{(end)}$.

\item  (Smooth bifurcate sphere.) Let
$
S_0=\pr(\II^{-}(\M^{(end)}))\cap\pr(\II^+(\M^{(end)})).
$
We assume that $S_0\subseteq\Sigma_0$ and $S_0$ is an imbedded $2$-sphere which agrees with the sphere of radius $1$ in $\mathbb{R}^3$ under the identification of $\Sigma_0$ with $\{x\in\mathbb{R}^3:|x|>1/2\}$. Furthermore, we assume that there is a neighborhood $\mathbf{O}$ of $S_0$ in $\mathbf{M}$ such that the sets
\begin{equation*}
\HH^+=\mathbf{O}\cap \pr(\II^{-}(\M^{(end)})\quad \text{ and }\quad \HH^-=\mathbf{O}\cap \pr(\II^{+}(\M^{(end)})
\end{equation*}
are smooth imbedded hypersurfaces diffeomorphic to $S_0\times(-1,1)$, We assume that these hypersurfaces are null, non-expanding\footnote{A null hypersurface is said to be non-expanding if the trace of its null second fundamental form vanishes identically.}, and intersect transversally in $S_0$.
\medskip
\item (Tangency at Horizon)        Finally, we assume that the vector-field $\T$ is tangent to both hypersurfaces $\HH^+=\mathbf{O}\cap \delta(\II^{-}(\M^{(end)}))$ and $\HH^-=\mathbf{O}\cap \delta(\II^{+}(\M^{(end)}))$.

\end{enumerate}
\begin{definition}
A space-time verifying the above assumptions   will be   called a  regular, nondegenerate  stationary  vacuum spacetime. 
\end{definition}  
\begin{remark}
Note that  the definition pre-supposes    the presence of a  unique connected horizon.   
\end{remark}

\subsection{Kerr spacetime}
In Boyer-Lindquist coordinates the Kerr metric 
takes the form,
\bea
ds^2=-\frac{q^2\Delta}{\Sigma^2}(dt)^2+\frac{\Sigma^2(\sin\theta)^2}{q^2}\Big(d\phi-\frac{2aMr}{\Sigma^2}dt\Big)^2+\frac{\rh^2}{\Delta}(dr)^2+q^2(d\theta)^2,
\label{metric-Kerr}
\eea
where, 
\beaa
q^2&=&r^2+a^2\cos^2\th,\qquad
\De= r^2+a^2-2Mr,\qquad
\Sigma^2=(r^2+a^2)q^2+2Mra^2(\sin\theta)^2.
\eeaa
On the horizon we have $r=r_+:=M+\sqrt{M^2-a^2}$ and $\De=0$.   The domain of outer communication $\mathbf{E}$ is given by $r> r_+$. 
One can show that  the   complex Ernst potential $\si$ and the complex scalar
 $\FF^2$  associated to the Killing vectorfield $\T=\pr_t$   are  given by
\bea
\si&=&1-\frac{2M}{r+ia\cos\th},\qquad \FF^2=-\frac{4M^2}{(r+ia\cos\th)^4}.
\eea
Thus,
\bea
-4M^2\FF^2&=&(1-\si)^4
\eea
everywhere in the exterior region. 
Writing $y+iz:=(1-\si)^{-1}$ we  observe that,
\beaa
y=\frac{r}{2M}\ge \frac{r_+}{2M}>\frac 12.
\eeaa
everywhere in the exterior region.

\subsection{Mars-Simon}   In \cite{Ma1}   M.  Mars   gave a  very useful  local  characterization of  the    Kerr  family
 in terms of  the vanishing of  complex  4-tensor  $\SS$, called the Mars-Simon  tensor.  In other words 
 $\SS$ plays the same role  in detecting    a  Kerr   spacetime as  the Riemann curvature tensor plays  in 
  detecting      flat space.    
  
  \begin{definition} 
  Given a  stationary spacetime    with Killing field   $\T$ and associated    Ernst potential $\si$, 
    we define  the Mars-Simon tensor\footnote{in regions where $\si\neq 1$.},
    \beaa
    \SS_{\a\b \mu\nu}:=\RR_{\a\b \mu\nu}+6(1-\si)^{-1} \QQ_{\a\b \mu\nu}.
    \eeaa 
    where,
    \beaa
    \QQ_{\a\b \mu\nu}:=\FF_{\a\b }\FF_{\mu\nu}-\frac{1}{3}\FF^2\II_{\a\b \mu\nu}
    \eeaa
    and,
    \beaa
    \II_{\a\b \mu\nu}:=(\g_{\a\mu}\g_{\b\nu}-\g_{\a \nu}\g_{\b\mu}+i\in_{\a\b \mu\nu})/4
    \eeaa
    
  \end{definition}
\begin{remark}
Note that $\RR, \QQ, \SS$ are all   complex, self dual    Weyl fields  in the sense defined above. 
\end{remark} 

Here is  an important     property of  $\SS$,  derived and made  use of  in  \cite{Ion-K2}.    
\begin{proposition}
\label{prop:BianchiS}
The tensor $\SS$  verifies the  equation,   with $h=(1-\si)$,
\beaa
\D^\a\SS_{\a\b\mu\nu}
&=&-6h^{-1}\T^\si\SS_{\si\rho\ga\de}\big[\FF_\b\,^\rho\, \de_\mu^\ga\, \de_\nu^\de+\frac 2 3 
\II^{\rho}\,_{\b\mu\nu}\FF^{\ga\de}\big]
\eeaa

\end{proposition}
 We   give a complete proof of the proposition in appendix \ref{sect:B2}. As a corollary we derive,
\begin{corollary}
The tensor $\SS$ verifies   a covariant wave equation of the form,
\bea
\label{wave-SS}
\square \SS=\MM(\D\SS, \SS)
\eea

\end{corollary}
\subsection{A Maxwell   System}
In the appendix we also   derive a Maxwell type   equation    for the following  slightly   modified 
 version of   the Mars-Simon tensor,
 \beaa
  \SS_{\a\b \mu\nu}:=\RR_{\a\b \mu\nu}+6 h^{-1} \QQ_{\a\b \mu\nu}.
 \eeaa
 where, for some constant\footnote{The precise constant in Kerr is   $   C=(4M^2)^{1/4} $ } $C$,
 \bea
 h&:=&C(- \FF^2)^{1/4}
 \eea
\begin{proposition} 
\label{prop-Max'}
The self-dual complex $2$-form $\HH_{\a\b}:=h^{-3}\SS_{\a\b \mu\nu}\FF^{\mu\nu}$,  
 verifies the  Maxwell equations,
 \begin{equation}\label{eq:Max'}
\D^\a\HH_{\a\b}= -h^{-3}\T^\si(\SS\cdot\SS)_{\si\b} -3h^{-1}E^\rho\HH_{\rho\b}, \qquad
\end{equation}
where,
\beaa
(\SS\cdot\SS)_{\si\b}=\SS_\b\,^{\rho\mu\nu}\SS_{\si\rho \mu\nu } 
\eeaa
and,
\bea
 E_\rho :=\si_\rho +\D_\rho  h=-\frac 1 2  C^4 \T ^\si\HH_{\si\rho}
 \eea
 \end{proposition}
\begin{remark} 
\label{rem:Max-linear} 
Remark that  the right hand side of \eqref{eq:Max'}    is  quadratic  in  $\SS$ and thus,  if $\SS=O(\ep)$,  sufficiently  small,
we can  ignore it in a first approximation   and derive the linearized equation 
\bea
\label{eq:Max-lin}
\D^\a\HH^{(lin)}_{\a\b}&=&0.
\eea

\end{remark}
\subsection{Rigidity  results based on $\SS$.}
\begin{theorem}[Ionescu-Klainerman \cite{Ion-K2}]
\label{Thm:Gl1}
Assume  that   $(\M, \g)$  is regular, nondegenerate  stationary  vacuum spacetime.  Assume also  that 
the following conditions  are  verified 

\begin{equation}\label{Main-Cond1}
-4M^2\FF^2=(1-\si)^4\quad\text{ on }S_0,
\end{equation}
and
\begin{equation}\label{Main-Cond2}
\qquad \qquad\qquad\Re\big((1-\si)^{-1}\big)>1/2\quad\,\text{  at some point on }S_0.
\end{equation}
Then $(\M, \g)$ is  isometric to the domain of outer communication of  of a Kerr space-time with mass $M$ and $0<a<M$.
\end{theorem}
The proof of the theorem is based on the following ingredients:
\begin{enumerate}
\item  Assumption \eqref{Main-Cond1}   is  used to show that $\SS$ vanishes  along the horizon. 
\item  Due to the non-degeneracy of the horizon  one can check that  the null  convexity   condition 
is verified at all points of the   bifurcate sphere.   Thus  our   unique continuation  results applied to  equation  \eqref{wave-SS} 
  can be applied  to prove that $\SS$ vanishes in a neighborhood of the horizon.  
    \item    This is the key step.   Define  functions $y, z$ such that    $ y+i z= (1-\si)^{-1} $.  Note that in    the particular case  of  $\KK(a, M)$ they are 
  $y=   (2M)^{-1} r, z=(2M)^{-1} \cos \th$.  Use the vanishing or $\SS$  to show that the level  set of $y$ 
      define a regular    foliation  of the  entire   domain of outer communication  and   verify the strong   $\T$-null convexity condition. 
      We can thus  apply   our unique continuation results  to     \eqref{wave-SS}  to deduce  that $\SS$ vanishes everywhere.
      \item According to Mars theorem \cite{Ma1}  we conclude that      our space-time is  isometric to $\KK(a, M)$, with $a, M$ determined
      from the asymptotic conditions of the metric $\g$.

\end{enumerate}

We now  state our second main theorem. Roughly the theorem  shows that any  stationary spacetime close  to  a   non-extremal  Kerr solution 
 $\KK(a, M)$, $|a|<M$,   must be     a  non-extremal  Kerr solution. The closeness  to Kerr   is expressed in terms of the  smallness of the  Mars-Simon tensor $\SS$.
\begin{theorem}[Alexakis-Ionescu-Klainerman \cite{AlIoKl2}]
\label{Thm:Gl2}
Assume  that   $(\M, \g)$  is regular, nondegenerate  stationary  vacuum spacetime. 
Replace assumption  \eqref{Main-Cond1} in   theorem 
by the condition,

\begin{equation}\label{Main-Cond0}
|(1-\sigma)\SS(\T,E_\al,E_\be,E_\ga)|\leq \varep \,\,\text{ on }\Sigma^0\cap\E ,
\end{equation}
for some sufficiently small constant  $ \varep$
(depending only on     our regularity  assumption on the  metric $\g$)  where $E_0, E_1, E_2, E_3$
is a  fixed  orthonormal frame along $\Si_0$ with $E_0$ the future unit normal.   Then, if $\varep$
 is sufficiently small,   the entire domain of communication $\E$ is isometric to  the exterior 
   region of  a Kerr solution  $\KK(a, M)$.  

\end{theorem}
The proof of the theorem is based on the following ideas.
\begin{enumerate}
\item A  simple  argument, due   to Hawking\footnote{Hawking's original argument     also applies to 
 degenerate horizons. In the case   of a non-degenerate horizon, assumed here,  
    the  proof is  completely  trivial.}   ,   shows that  one can construct a second Killing vectorfield  $\K$  on  
  the horizon  $\HH^+\cup\HH^{-}$,
 with $\K$   tangent to the generators.  
 \item   We check that the strict-null convexity condition  is   verified at all point of  the bifurcate sphere $S_0$
 and extend  $\K$  in a full  neighborhood of $S_0$.  Moreover    $\K$ commutes with $\T$, $[\K,\T]=0$. 
     \item Introduce the coordinates $y, z$  such that $y+iz=(1-\si)^{-1}$  and show, using the smallness assumption on $\SS$,
     that  $y$ verifies the strict $\T$- null-convexity  condition.  
     \item   Extend $\K$ everywhere in $\E$  as a Killing vectorfield such that $[\K,\T]=0$ and find a combination  $\Z$ of $\T, \K$ which   has closed  
    orbits.
    \item Use  the Carter-Robinson theorem to conclude  that   $\E$ is isometric to the exterior domain of a Kerr solution.

\end{enumerate}

\begin{remark}
Theorem \ref{Thm:Gl2} has been significantly strengthened  by Wong and Yu in \cite{Wong-Yu} in which they show,  by 
 a  clever application of the mountain pass lemma,     that  the assumption of a  connected  horizon, implicit  in    both    
    theorems   \ref{Thm:Gl1}, \ref{Thm:Gl2}, is unnecessary.
         \end{remark}

\begin{remark}
The reliance on the Carter-Robinson theorem  in the last step of the proof is somewhat unsatisfactory
since we are   in a  small $\SS$ regime. In fact the authors      believe that an alternative argument
 can be given  relying on proposition \ref{prop-Max'}  and the study of stationary solutions to the linearized      
 system  \eqref{eq:Max-lin}.
 \end{remark}

\subsection{   Third  rigidity   result}
In this  section  we   review    a   recent    black hole rigidity result 
for slowly rotating stationary solutions of the Einstein  vacuum  equations. 
 The result states    that the domain of outer communications of any  stationary  vacuum black hole\footnote{verifying   the assumptions   
 in  subsection \ref{sect:stationary}}
with the stationary Killing vector-field $\T$ being small on the 
bifurcation sphere of the horizon must be isometric to the domain of outer communications of a Kerr solution $K(a,M)$ with 
small  $a$. More precisely,
\begin{theorem}[Alexakis-Ionescu-Klainerman \cite{AlIoKl3}]  \label{thm:Gl3}
Assume  that   $(\M, \g)$  is regular, nondegenerate  stationary  vacuum spacetime,    as  in subsection  \ref{sect:stationary}. 
Assume in addition that   there exists a    regular    maximal    hypersurface $\Sigma_1$ passing through the bifurcation  sphere 
 and that 
\begin{equation}\label{Tsmall}
\|\g(\T,\T)\|_{L^\infty(\SS_0)}<\epsilon,
\end{equation}
where $\epsilon$ is a sufficiently small constant\footnote{We note that the smallness depends  on the entire geometry of  $(\M, \g)$, in particular on  its ADM mass  $M$.}. 
Then  $(\M,\g)$ is stationary and axially symmetric,  
thus, in view of the Carter-Robinson theorem,     isometric 
to a Kerr spacetime $\KK(a, m) $   with small  $a$. 
\end{theorem}

  This  result should  be compared with     that  stated in    theorem \ref{Thm:Gl2}               in which  rigidity 
was proved, for the entire range $0\le a <M$, under a  global smallness  assumption on the Mars-Simon tensor associated
 to  the stationary  space-time.    We recall   that        the proof  of    theorem \ref{Thm:Gl2}    rested on  the following 
  ingredients:
\begin{enumerate}
\item  An unconditional local rigidity result,  discussed in section 2,   according to which  a second, rotational Killing vector-field  $\Z$ can be constructed in a small neighborhood of the 
 bifurcate  sphere of the horizon.
 \item An extension argument  for the  Killing vector-field $\Z$ based on    a   global  foliation of the  space-time  with   $\T$- conditional pseudo-convex    hypersurfaces.   
 The crucial $\T$- conditional pseudo-convexity  condition is ensured by the  assumed smallness of the  Mars-Simon tensor.
 \item Once $\Z$ is  globally extended, and thus the space-time is shown  to be  both stationary  and axisymmetric, one can appeal to the classical 
 Carter-Robinson theorem  to conclude the desired  rigidity. 
\end{enumerate} 

Theorem \ref{thm:Gl3}   is still based on the  first and third ingredients above but replaces  the second  one with  a new ingredient 
inspired  from    the classical work of Sudarsky and Wald \cite{Su-Wa} (see also \cite{Ca2})      on the staticity of stationary, axially 
symmetric,  black hole solutions 
with zero angular momentum.  The Sudarski-Wald   result  was based on a   simple  integral formula linking the total extrinsic curvature of   a regular maximal  hypersurface      $\Sigma$ imbedded   in    the space-time and passing through the bifurcate sphere,      with the angular momentum of the horizon.    It can be easily shown\footnote{This step 
 is based on   the additional   assumption of axial symmetry.}           that    zero ADM angular momentum
implies  vanishing angular momentum of the horizon and thus, in view of the above mentioned  formula,   the maximal hyper-surface has to be totally geodesic.  This then implies   the desired conclusion of \cite{Su-Wa}, 
 i.e the space-time  is static. The main observation  in the proof of theorem    \ref{thm:Gl3}     is that    a simple smallness   assumption of $\T$ on the bifurcate sphere\footnote{ This is  equivalent 
with  a   small angular momentum assumption   on the horizon.  It remains open whether this   condition can be replaced with 
  a smallness assumption  of the ADM angular momentum.}  implies  the smallness 
of  the total curvature of the  maximal  hypersurface. This can then be combined   with a  simple application of the  classical    Hopf   Lemma 
to conclude that the entire ergo-region of the black hole     can be covered  by   the   local    neighborhood of the horizon in which the 
second,  rotational, Killing vector-field  $\Z$ has been extended, according to step (1) above.   Away  from the ergo-region $\T$ is time-like   and thus   
 $\T$-conditional  pseudo-convexity  is automatically  satisfied. Thus,  the second Killing  vector-field   $\Z$ can   be easily extended      to the entire space-time  by the results  discussed in section 2.
 
 \subsection{ Einstein-Maxwell case}
 The results   of theorems    \ref{Thm:Gl1},  \ref{Thm:Gl2}       presented in this section have been extended to the  Einstein-Maxwell equations
   by W. Wong and P. Yu.             The analogue of the Mars Simon  tensor  was  discovered by Wong in \cite{Wong1}.   It consists of a pair of tensors,  
    one     related to the  curvature tensor and   the   second related to the Maxwell  field.   A Kerr-Newman solution is characterized by their simultaneous vanishing.   A slight modification of the    pair appears in \cite{Wong-Yu}.  The applications  to the rigidity problem appear in  \cite{Wong-Yu} as well as      \cite{Wong2} and \cite{Yu}.
   
\section{Conclusions}
Despite      statements to the contrary made   often  in   physics literature,   the rigidity  conjecture      remains wide open.   
     Though        a lot of progress was   made in the last ten years,       the  full scope of the conjecture   remains out of reach.
       The  global  results presented in this survey  are  mostly   limited    to  perturbative regimes.      Under  somewhat different   assumptions 
        they all assert that     that stationary solutions  closed   to a non-extremal Kerr   must be  isometric   to  a  a non-extremal Kerr    whose 
         parameters $a, M$ are determined  by  their    ADM mass   and angular momentum.   Despite  their   limitations 
       they offer  however a perspective   of  what one might expect  to encounter     in the general case.   
       To start with,    the    results illustrate 
        the important role played by null geodesics  perpendicular  to  the   stationary  Killing vectorfield $\T$.     Based on the experience
          we have accumulated so far,  we conjecture,  together with our collaborator  Spyros Alexakis,    the following   general conjecture.
        
        \medskip
          
    \noindent      {\bf Conjecture}[Alexakis-Ionescu-Klainerman].  \textit{ Any  asymptotically flat,  regular,   stationary    vacuum  solution, as  
     as  in subsection  \ref{sect:stationary}, which   admits no trapped null geodesics    perpendicular   to $\T$  must be 
     isometric to the exterior   part of a  non-extremal Kerr solution. }

\medskip
 All the three global results     discussed in the survey      are based on   the  fact
      that small, stationary perturbations  of  a     non-extremal   Kerr spacetime\footnote{Note however that the Kerr 
    family admits plenty of trapped  null geodesics.}    verifies 
       the hypothesis of the conjecture.
    It is conceivable   that  the conjecture can be proved  with current techniques,  based on unique continuation  methods. The conjecture leaves 
    however open the   question whether such     null geodesics   can be ruled  out in general.   It thus   illustrates an important aspect of the  general  case,
     namely      the fact that    we cannot  hope to prove the    full rigidity     conjecture  based  only on a continuation argument starting  form the horizon.
    Indeed such an argument may not distinguish  between   the given stationary     Killing vectorfield $\T$  and any other  possible Killing   vectorfield, such
     as  $\T+c \Z$  in  Kerr.   While, in Kerr, there are no trapped null geodesics perpendicular to $\T$  there are plenty of those perpendicular\footnote{In fact for any given trapped null geodesic  we can find a constant   $c$  such that  $\T+c \Z$  is perpendicular to it.  } to 
     $\T+c\Z$.  Thus  a full proof of the  rigidity  conjecture  must   rely    on  global properties of the space-time.

\begin{appendix}
\section{Proof of proposition \ref{prop:Carleman1}}
We    first restate  the proposition in a general   setting  of  an arbitrary Lorentzian  manifold  $(M, g)$, a domain 
$O\subset M$, $p\in\pr O$ and    $h$ a    defining, nondgenerate,  function for  $\pr O$         in a full neighborhood $ U_1$  of 
  $p$, i.e. $h<0$ in    $O\cap U_1$  and  $h=0$   on      $\pr O\cap U_1 $.   Moreover we assume that  $h$
  verifies      the   condition  \eqref{null-convex-strong}  in $U_1$. 
  \bea
\label{null-convex-strong'}
\begin{cases}
&\quad \qquad \qquad \qquad \qquad \qquad \qquad| d  h|  \ge  M^{-1}\\
&Y^\a Y ^\b  \left( \mu    g_{\a \b}- D_\a   D_\b h\right)+M|Y (h )|^2\,\, \geq  M^{-1}|Y|^2,
\end{cases}
\eea
uniformly, at all points of $U$, with $|Y|^2= (Y ^0)^2+(Y^1)^2+\ldots+(Y^d)^2$..

    Let  $U_\ep$ be small  neighborhoods of 
 $p$ such that $|h|\le  2^{-1}\ep$    in $U_\ep$ and define  the weight functions,
  $f_\ep: U_\ep \longrightarrow \RRR$
 \beaa
 f_\ep :=\log (\ep+h+e_p)
 \eeaa
 where $e_p$ is a small  perturbation of  $\ep+h$. More precisely, we say that  $e_\ep$ is a negligible perturbation 
 if 
 \beaa
 \sup_{U\ep}  |D^j e_\ep|\le \ep^2     \qquad \mbox{for}  \,\, j=0,1,2.
 \eeaa
 In particular   the    weights  $f_\ep$     verify \eqref{null-convex-strong} in $U_\ep$, uniformly in $\ep>0$.
 Also, uniformly in $U_\ep$,
 \beaa
 |Df_\ep|\le C\ep^{-1}
 \eeaa

\begin{proposition}
\label{prop:Carleman1}
 If $f_\ep$ are as above,     there exists  a sufficiently  small     $\ep>0$  and a  large constant $C_\ep>0$ 
   such that, for all $\phi \in C_0^2(U_\ep)$ and all sufficiently large  $\la>0$,
     \begin{equation}\label{ca1}
\lambda \cdot \|e^{-\lambda f_\ep }\cdot \phi\|_{L^2}+
\|e^{-\lambda f_\ep }\cdot D \phi\|_{L^2}\leq C_\ep    \lambda^{-1/2}\cdot \|e^{-\lambda f_\ep}\cdot\square_\g \phi\|_{L^2},
\end{equation}

\end{proposition}
\begin{remark}
Note that $C_\ep$    denotes a constant which depends only on  the small parameter $\ep$ but not on $\la$.   Throughout  the proof below we shrink $\ep>0$  whenever necessary and enlarge the constant  $C_\ep$.

\end{remark}
\begin{proof}
 We first  fix $\ep>0$.     
  Since all derivatives of $f=f_\ep $ are bounded on  $U=U_\ep$      it suffixes to prove ( with a different $C_\ep$ !), 
\begin{equation}\label{car1}
\lambda \cdot \|e^{-\lambda f}\cdot \phi\|_{L^2}+\| D(e^{-\lambda  f}\cdot \phi\big)\|_{L^2}\leq C_\ep\lambda ^{-1/2}\cdot \|e^{-\lambda  f}\cdot\square_\g \phi\|_{L^2}.
\end{equation}
To prove  estimate \eqref{car1} we start by setting,
\bea
\phi=e^{\la f}\psi
\eea
Observe that,
\beaa
e^{-\la f}\square (e^{\la f}\psi)&=&\square \psi+\la(2 D^\b f D_\b \psi+\square f\psi)+\la^2( D^\b f D_\b f)\psi\\
&=& L\psi+\square  f \psi
\eeaa
where,
\beaa
L\psi=\square\psi+ 2\la X(\psi) +\la^2 G\psi,\qquad   X=D^\a f\pr_\a,\qquad G=D^\be f  D_\b f.
\eeaa
Thus estimate  \eqref{car1}
follows from,
\bea
\label{car2}
\la\|\psi\|_{L^2}+ \|D\psi\|_{L^2}\le C_\ep \la^{-1/2}
 \|L\psi\|_{L^2},
\eea
Recall  the energy moment tensor of $\square=\square_g$,
\beaa
 Q_{\mu\nu }=D_\mu \psi D_\nu \psi-\frac{1}{2} g_{\mu\nu }( D^\si \psi  D_\si \psi).
\eeaa
Given a vectorfield  $X$  and  a  scalar function  $w$   we   define $P_\mu =P_\mu[X, w ]$
\beaa
P_\mu : &=& Q_{\mu \nu } X^\nu  - w  \phi\pr_\mu  \phi  +\frac 1 2 \pr_\mu w \phi^2
\eeaa
\begin{lemma}
The   one form $P_\mu =P_\mu[X, w]$  verifies the identity,
\bea
D^\mu P_\mu&=&(X(\psi) - w \psi )\square \psi +\frac{1}{2}Q_{\mu\nu}{}^{(X)}\pi^{\mu\nu}+\frac{1}{2}\square w\psi^2
 -w g(d\psi, d\psi)  \label{Main-pointwise}
\eea

\end{lemma}
In our case we have $\square\psi =  L\psi -    2\la X(\psi) -\la^2 G\psi$.   Hence,
\beaa
D^\mu P_\mu&=&(X(\psi)-    w \psi) \left(L  \psi   -   2\la  X(\psi)      -\la^2 G\psi    \right)   \\
 &+&\frac{1}{2}Q_{\mu\nu}    \piX^{\mu\nu}
-   w   D^\mu \psi D_\mu \psi+   \frac 1 2   \square_g  w |\psi|^2
\eeaa
or, 
\beaa
( Div \, P) +\la  |X(\psi)|^2   -\frac{1}{2}Q_{\mu\nu}    \piX^{\mu\nu}+  w  D^\mu \psi D_\mu \psi&=& E
\eeaa
where,
\beaa
E&=& (X(\psi)-    w \psi) \left(L  \psi   -   2\la  X(\psi)      -\la^2 G\psi    \right) +\la |X(\psi)|^2 +    \frac 1 2   \square_g  w |\psi|^2\\
&=&  (X(\psi)-    w \psi)    L  \psi + (X(\psi)-    w \psi) \left[   -\la (X(\psi)- w \psi)  -\la (X(\psi)+ w \psi)     -\la^2 G\psi    \right)\\
&+&\la |X(\psi)|^2 +    \frac 1 2   \square_\g  w |\psi|^2\\
&=&(X(\psi)-    w \psi)    L  \psi - \la  (X(\psi)-    w \psi)^2-\la\left( |X(\psi)|^2-w^2\psi^2\right)+\la |X(\psi)|^2 +    \frac 1 2   \square_\g  w |\psi|^2\\\
&-&\la^2 G \psi(X(\psi)-    w \psi) \\
&=&(X(\psi)-    w \psi)    L  \psi - \la  (X(\psi)-    w \psi)^2+|\psi|^2 \left( \la w ^2  +\frac 1 2   \square_g  w  \right)-\la^2 G \psi(X(\psi)-    w \psi) 
\eeaa
  Note that,
  \beaa
  G\psi (X(\psi)-   w \psi) =\frac 1 2 G X^\mu  D_\mu  (\psi^2) -w G\psi^2 = D_\mu \left( \frac 1 2 G   X^\mu \psi^2  \right) -\frac 1 2 \psi^2\left[    D_\mu ( G   X^\mu ) +2 G w      \right] 
  \eeaa
Thus,
\begin{lemma}
We have the point wise identity,
  \beaa
   D^\mu P'_\mu+   \la |X(\psi)|^2    -\frac{1}{2}Q_{\mu\nu}    \piX^{\mu\nu}+ w  D^\mu \psi D_\mu \psi  -\frac{\la^2}{2} \psi^2\left[    \D_\mu ( G   X^\mu ) +2 G w      \right]&= & E'
  \eeaa
  where,
  \beaa
  P'_\mu=P_\mu+ \frac 1 2\la^2  G X_\mu \psi^2 
  \eeaa
and,
\beaa
E'&=&(X(\psi)-    w \psi)    L  \psi - \la  (X(\psi)-    w \psi)^2
+|\psi|^2 \left( \la w ^2  +\frac 1 2   \square_\g  w  \right)
\eeaa
\end{lemma}
 Now,
  \beaa
   (X(\psi)-  w \psi)   L  \psi\le   \la^{-1}  | L\psi | ^2+\la |  X(\psi)-  w \psi|^2 
  \eeaa
Hence,
\beaa
E'\le   \la^{-1}  | L\psi | ^2+|\psi|^2 \left( \la w ^2  +\frac 1 2   \square_\g  w  \right)
\eeaa
 Since by   integration  $D^\a  P'_\a$ disappears,
 it suffices   to check that the desired   inequality
 \beaa
\la^2\|\psi\|_{L^2}^2+\|D\psi\|_{L^2}^2 \le C_\ep \la^{-1}
 \|L\psi\|_{L^2}^2
\eeaa
 for       $\la $  sufficiently large,  follows     by integrating the following pointwise inequality,
 \beaa
&&  \la |X(\psi)|^2    -\frac{1}{2}Q_{\mu\nu}    \piX^{\mu\nu}+ w  D^\mu \psi D_\mu \psi  -\frac{\la^2}{2} \psi^2\left[    D_\mu ( G   X^\mu ) +2 G w      \right]\\
&&\le  
    \la^{-1}  | L\psi | ^2+|\psi|^2 \left( \la w ^2  +\frac 1 2   \square_\g  w  \right)
 \eeaa
It  thus suffices\footnote{Note indeed  that the remaining  term $|\psi|^2 \left( \la w ^2  +\frac 1 2   \square_\g  w  \right)$   on the right  of the  above inequality 
 is  lower order and can be easily absorb.
}   to prove the following two  inequalities, for $C_\ep$ sufficiently  large  and $\la $ large,
\beaa
 \la |X(\psi)|^2 -\frac{1}{2}Q_{\mu\nu}    \piX^{\mu\nu}+ w  D^\mu \psi D_\mu \psi\ge  C_\ep^{-1}|D\psi|^2 \\
  -\frac{1 }{2} \left[    D_\mu ( G   X^\mu ) +2 G w      \right]\ge C_\ep^{-1} 
\eeaa
Recalling  the definition of $X=D^\a f  \pr_\a $ and $G= D_a f D^\a f $ we write,
\beaa
 -\frac{1}{2}Q_{\mu\nu}    \piX^{\mu\nu}+ w  D^\mu \psi D_\mu \psi=  \left[ -D^\mu D^\nu  f  +\frac 1 2  g^{\mu\nu}  (2w+ \square f)  \right] D_\mu \psi D_\nu \psi 
\eeaa
\beaa
    D_\mu ( G   X^\mu ) +2 G w   &=&   X(G)   +  G     D_\mu  X^\mu  +2 G w   = D^\a  f D_\a  ( D^\b f D_\b f)+G( \square f +2 w) \\
    &=&2 D^\a f   D^\b f  D_ \a D_\b f  +  D^\a  f D_\a f ( \square f +2 w)
 \eeaa 
Hence, it suffices to show the inequalities  (with  $2w'=  \square f +2 w$ ),
\bea
\label{car3}
 \la |D^\a f D_\a \psi |^2  +\left[-D_ \a  D_\b   f  +  w'  g_{\a\b } \right]    D^\a \psi D^\b  \psi &  \ge  &  C_\ep^{-1}|D\psi|^2     \\ 
-    \left(   D_ \a D_\b f  + w '    g_{\a\b} \right)   D^\a  f    D^\b  f  &\ge&   C_\ep^{-1}
 \label{car4}
\eea
Now recall that    $f=f_\ep=\log(h+\ep)    $. Therefore,
\beaa
D_\a  f&=& (h+\ep)^{-1}  D_\a  h\\
D_\a D_\b  f&=& (h+\ep)^{-1}  D_\a   D_\b h-   (h+\ep)^{-2} D_\a h   D_\b h
\eeaa
The inequality 
\eqref{car3} becomes,
   \beaa
   (\la  +1)  (h+\ep)^{-1}     |D^\a  h D_\a \psi |^2  +\left[-           D_\a  D_\b   h      +  w   (h+\ep) g_{\a\b } \right]    D^\a \psi D^\b \psi  \ge      C_\ep^{-1 }     (h+\ep)    |D\psi|^2
   \eeaa
   which follows if,
   \beaa
  \frac{1}{2\ep}   \la |D^\a  h D_\a \psi |^2+\left[  w   (h+\ep) g_{\a\b} -       D_\a  D_\b   h\right)]      D^\a \psi D^\b \psi   \ge   2\ep   C_ \ep^{-1}  |D\psi|^2
   \eeaa
   for some $C_\ep$   large. 
     The inequality \eqref{car4} becomes,   
     \beaa
     -\left[ (h+\ep)^{-1} D_\a  D_\b h-(h+\ep)^{-2}   D_\a h D_\b h \right] D^\a h D^\b h - w' \g_{\a\b} D^\a h D^\b h \ge C_\ep^{-1} (h+\ep)^2
     \eeaa
     or,
     \beaa
     (h+\ep)^{-1} (D_\a h D^\a h)^2-    D_\a  D_\b h  D^\a h D^\b h -w' (h+\ep)     (D_\a h D^\a h)  \ge C_\ep^{-1} (h+\ep)^3
     \eeaa
     It thus suffices to   have,    with a slightly different $C_{\ep}$,
     \beaa
     \frac{1}{2 \ep}  (D_\a h D^\a h)^2 -    D_\a  D_\b h  D^\a h D^\b h  -w' (h+\ep)     (D_\a h D^\a h)  \ge    C_\ep^{-1}  \ep^3 
     \eeaa           
     It thus  remains to show  that,  with the right choice of $\ep$ sufficiently small     and $\la $ sufficiently large  we can verify the inequalities,
     \bea
      \frac{1}{2 \ep }   \la |D^\a  h D_\a \psi |^2+\left[  w'   (h+\ep) g_{\a\b} -       D_\a  D_\b   h\right)]   D^\a \psi D^\b \psi  & \ge &    2 \ep C_1   |D\psi|^2\label{car5}\\
 \frac{1}{2 \ep}  (D_\a h D^\a h)^2 -    D_\a  D_\b h  D^\a h D^\b h  -w' (h+\ep)     (D_\a h D^\a h) & \ge &   C_\ep^{-1}  \ep^3    \label{car6}
     \eea
          In view of the  quantitative    strict  null convexity condition  \eqref{null-convex-strong} we have,    for  a sufficiently large  $M$, $\mu\in[-M, M]$  and    for all vector fields $Y$  in  $U$,
  \bea
Y^\a       Y ^\b (\mu    g_{\al\be}- D_\al  D_\be h )+M|Y h |^2\geq M ^{-1}|Y|^2,\qquad   | D h|\ge   M^{-1}.   \label{null-convex-strong'} 
\eea
Hence
 with $Y^\a =D^\a \psi$, $\la \ge M $, 
 \beaa
D ^\a \psi D ^\b  \psi ( \mu g_{\a\b }- D_\al  D_\be h )+\la |D^\a \psi  D_\a  (h )|^2\geq  M  ^{-1}|D \psi|^2,
 \eeaa
    Therefore, choosing $ w'=(h+\ep)^{-1}  \mu $, i.e. $w=(h+\ep)^{-1}\mu - \frac 1 2\square   f_\ep$, we deduce that  \eqref{car5} holds true
 provided that  $\ep\ll M^{-1}$.  
     
     If $D^\a h D_\a h (p)\neq 0$,  we can also  find        $\ep$ sufficiently small such   \eqref{car6}   holds  uniformly 
     on $U=U_\ep$.  Once again we need 
   $\ep\ll M^{-1}$. 
   If     $D_\a h D^\a h=0$  holds  at $p$,  we take  $Y= D^\a h  \pr_\a $   in   \eqref{null-convex-strong'}    and   derive,  at $p$,
   \bea
   - D_\a D_\b h    D_\a h D_\b h (p)  \ge M ^{-1}| dh(p) |^2\label{car7}
   \eea
   On the other hand   the inequality \eqref{car6} becomes, at $p$.
   \beaa
   -     D_\a   D_\b h D^\a h D^\b h      &\ge &   2 C_\ep^{-1}\ep ^3
   \eeaa
   In view of \eqref{car7}   this last inequality  is satisfied  if,
   \beaa
    M ^{-1}| dh(p) |^2\ge    2 C_\ep^{-1}\ep ^3.
   \eeaa
   Since    $|dh(p)|\ge M^{-1}$ 
   we   need   $M^{-3 } \gg  C_\ep   \ep^3  $.
   Thus,  to have both inequalities  satisfied on $U_\ep$   we  need,
   \beaa
   \ep \ll  M^{-1} 
   \eeaa
   for a sufficiently large    $C_\ep$
In other words given $M$, such that \eqref{null-convex-strong'}  is verified,   we  first choose $\ep  \ll  M^{-1}$   and  then choose  
 $C_\ep$  sufficiently large.

\end{proof}

\section{Proof of Proposition \ref{prop:BianchiS}  }
\label{sect:B2}
We  give a slightly  modified  definition of the Mars-Simon tensor.  
 \bea
    \SS_{\a\b \mu\nu}:&=&\RR_{\a\b \mu\nu}+6 h  ^{-1} \QQ_{\a\b \mu\nu},\label{eq:gen-SS}\\
    \QQ_{\a\b \mu\nu}:&=&\FF_{\a\b }\FF_{\mu\nu}-\frac{1}{3}\FF^2\II_{\a\b \mu\nu},\nn\\
    \II_{\a\b \mu\nu}:&=&(\g_{\a\mu}\g_{\b\nu}-\g_{\a \nu}\g_{\b\mu}+i\in_{\a\b \mu\nu})/4\nn
    \eea
 with  an arbitrary   function     $h$.
\begin{proposition}
\label{prop:BianchiS}
The tensor $\SS$  verifies the  equation, 
\beaa
\D^\a\SS_{\a\b\mu\nu}
&=&-6h^{-1}\T^\si\SS_{\si\rho\ga\de}\big[\FF_\b\,^\rho\, \de_\mu^\ga\, \de_\nu^\de+\frac 2 3 
\II^{\rho}\,_{\b\mu\nu}\FF^{\ga\de}\big]  - 6 h^{-2}   E^\rho   \QQ_{\rho\b \mu\nu}
\eeaa
where,
\beaa
E_\rho=\D_\rho h+ \si_\rho
\eeaa
\end{proposition}
\begin{remark}
Note that proposition \ref{prop:BianchiS} is an immediate consequence      for   the special case  $h=(1-\si)$.
\end{remark}
\begin{proof} In view of          \eqref{eq:Ricci}   and the  definitions of $\SS$ and $\si$,
\beaa
\D_\a\FF_{\b\ga}&=&\T^\la \RR_{\la\a\b\ga}=   \T^\la \SS_{\la\a\b\ga}-    6h^{-1}   \T^\la \QQ_{\la\a\b\ga}         \\
&=&\T^\la \SS_{\la\a\b\ga}-
6h^{-1}  \T^\la\big(\FF_{\la\a}\FF_{\b\ga}-\frac 1 3 \FF^2 \II_{\la\a\b\ga}\big) \\
&=&\T^\la \SS_{\la\a\b\ga}-3h^{-1}\sigma_\a\FF_{\b\ga}+2h^{-1}\FF^2\T^\la \II_{\la\a\b\ga}
\eeaa
i.e.,
\bea
\label{eq:deriv-FF}
\D_\a\FF_{\b\ga}&=&2h^{-1}\FF^2\T^\la \II_{\la\a\b\ga}-3h^{-1}\sigma_\a\FF_{\b\ga}+\T^\la\SS_{\la\a\b\ga}
\eea
Thus,
\beaa
\D_\a[h^{-3}\FF_{\b\ga}]&=&h^{-3}\left[2h^{-1}\FF^2\T ^\la \II_{\la\a\b\ga}-3h^{-1}\sigma_\a\FF_{\b\ga}+ \T^\la\SS_{\la\a\b\ga} \right]-
3h^{-4} \D_\a h \FF_{\b\ga}\\
&=&2h^{-4}\FF^2\T^\si\II_{\si\a\b\ga}-3h^{-4}(\si_\a+\D_\a h)\FF_{\b\ga}+h^{-3}\T^\la \SS_{\la\a\b\ga}\\
&=&2h^{-4}\FF^2\T^\si\II_{\si\a\b\ga}-3h^{-4}  E_\a \FF_{\b\ga}+h^{-3}\T^\la \SS_{\la\a\b\ga}
\eeaa
We record this  result  for future reference,
\bea
\D_\a[h^{-3}\FF_{\b\ga}]&=&2h^{-4}\FF^2\T^\si\II_{\si\a\b\ga}-3h^{-4}  E_\a \FF_{\b\ga}+h^{-3}\T^\la \SS_{\la\a\b\ga}\label{eq:interm-H}
\eea

Since  $ \II_{\la\a\b\ga}\FF^{\b\ga}=\FF_{\la\a}$,
\beaa
\FF^{\b\ga}D_\a\FF_{\b\ga}&=&-3h^{-1}\sigma_\a\FF^2+2h^{-1}\FF^2\T ^\la\FF_{\la\a}+\T^\la \SS_{\la\a\b\ga}\FF^{\b\ga}\\
&=&-2h^{-1}\sigma_\a\FF^2+\T^\la \SS_{\la\a\b\ga}\FF^{\b\ga}
\eeaa
Thus,
\bea
\D_\a(\FF^2)&=&-4h^{-1}\sigma_\a\FF^2+2\T ^\la \SS_{\la\a\b\ga}\FF^{\b\ga}\label{ge4}
\eea
We now calculate, using  \eqref{eq:deriv-FF}, \eqref{ge4}          and   \eqref{eq:Max2'},
\beaa
\D^\a \QQ_{\a\b\mu\nu}&=&\g^{\rho\a}\D_\rho\big(\FF_{\a\b}\FF_{\mu\nu}-(1/3)\FF^2\II_{\a\b\mu\nu}\big)\\
&=&\g^{\rho\a}\FF_{\a\b}\D_\rho\FF_{\mu\nu}+\g^{\rho\a}\FF_{\mu\nu}\D_\rho\FF_{\a\b}-(1/3)\g^{\rho\a}\II_{\a\b\mu\nu}\D_\rho\FF^2\\
&=& \g^{\rho\a}\FF_{\a\b}\D_\rho\FF_{\mu\nu}-(1/3)\g^{\rho\a}\II_{\a\b\mu\nu}\D_\rho\FF^2\\
&=&-\FF_\b\,^\rho
[2h^{-1}\FF^2\T^\si\II_{\si\rho\mu\nu}-3h^{-1}\sigma_\rho\FF_{\mu\nu}+\T^\si\SS_{\si\rho\mu\nu}]\\
&-&(1/3)\II_{\a\b\mu\nu}[-4h^{-1}\sigma^\a\FF^2+2\g^{\a\rho}( \T ^\si\SS_{\si\rho \ga\de})\FF^{\ga\de}]\\
&=&-2h^{-1}\FF^2\FF_\b^{\,\,\,\,\rho}\T^\si\II_{\si\rho\mu\nu}+3h^{-1}\FF_{\b\rho}\sigma^\rho\FF_{\mu\nu}+
(4/3)h^{-1}\FF^2\II_{\a\b\mu\nu}\sigma^\a\\
&-&\big[\FF_\b\,^\rho  \T^\si\SS_{\si\rho\mu\nu}+(2/3)\II_{\a\b\mu\nu}\g^{\a\rho}\FF^{\ga\de}\T^\si\SS_{\si\rho \ga\de}\big].
\eeaa
We deduce,
\beaa
\D^\a \QQ_{\a\b\mu\nu}&=&h^{-1} A_{\b\mu\nu}+B_{\b\mu\nu}\\
A_{\b\mu\nu}&=&-2\FF^2 \FF_\b\,^\rho \T^\si\II_{\si\rho\mu\nu}
+3\FF_{\b\rho}\sigma^\rho\FF_{\mu\nu}+(4/3)\FF^2\II_{\rho\b\mu\nu}\sigma^\rho.\\
B_{\b\mu\nu}&=&-\big[\FF_\b \,^\rho \T^\si\SS_{\si\rho\mu\nu}+(2/3)\II_{\a\b\mu\nu}\g^{\a\rho}\FF^{\ga\de}\T^\si\SS_{\si\rho \ga\de}\big]\\
\eeaa
Recalling the definition of $\QQ$ we derive, 
\beaa
A_{\b\mu\nu}&=&-2\FF^2\FF_\b\,^\rho  \T^\si\II_{\si\rho\mu\nu}+2\FF_{\b\rho}\sigma^\rho\FF_{\mu\nu}+\FF^2\II_{\rho\b\mu\nu}\sigma^\rho   -  \si^\rho\QQ_{\rho\b\mu\nu}\\
&=&K_{\b\mu\nu}-  \si^\rho\QQ_{\rho\b\mu\nu}
\eeaa
where, 
\beaa
K_{\b\mu\nu}&=&-2\FF^2\FF_\b\, ^\rho \T^\si\II_{\si\rho\mu\nu}+2\FF_{\b\rho}\sigma^\rho\FF_{\mu\nu}+\FF^2\II_{\rho\b\mu\nu}\sigma^\rho 
\eeaa
Making use of the identity,
 \bea
 \FF_{\mu}^{\,\,\,\si}\II_{\nu\sigma \a\b} +\FF_{\nu}^{\,\,\,\si} \II_{\mu\sigma\a\b}=
 \frac{1}{2}\g_{\mu\nu}\FF^{\ga\de}\II_{\ga\de\a\b}     =     \frac{1}{2}\g_{\mu\nu}  \FF_{\a\b}    .\label{b8.2''} 
\eea
and  $\si_\rho=2\T^\si\FF_{\si\rho}$ we derive,
\beaa
K_{\b\mu\nu}&=&-2\FF^2\FF_\b\,^\rho \T^\si\II_{\si\rho\mu\nu}+\FF^2\II_{\rho\b\mu\nu}\sigma^\rho+2\FF_{\b\rho}\sigma^\rho\FF_{\mu\nu}\\
&=&   -2\FF^2\big(  \FF_\b\,^\rho \T^\si\II_{\si\rho\mu\nu}-\II_{\rho\b\mu\nu} \T^\si\FF_{\si\rho} \big) +4\T^\si\FF_{\si}\,^\rho \FF_{\b\rho}\FF_{\mu\nu}\\
&=&-2\FF^2\T^\si(\FF_\b\,^\rho \II_{\si\rho\mu\nu}+\FF_\si\,^\rho \II_{\b\rho\mu\nu})+4\T^\si\FF_{\si}\,^\rho \FF_{\b\rho}\FF_{\mu\nu}\\
&=&-2\FF^2\T^\si\cdot (1/2)\g_{\b \si}\FF_{\mu\nu}+4\T^\si\FF_{\mu\nu}\cdot (1/4)\g_{\si\b}\FF^2\\
&=&0.
\eeaa
Consequently,
\beaa
\D^\a \QQ_{\a\b \mu\nu}&=&-h^{-1}\sigma^\rho\QQ_{\rho\b\mu\nu} +B_{\b\mu\nu}
\eeaa
  from which we deduce, recalling  $\si_\rho+\D_\rho h=E_\rho$
  \beaa
  \D^\a\big(h^{-1}  \QQ_{\a\b \mu\nu})&=&-h^{-2} \sigma^\rho\QQ_{\rho\b\mu\nu}+h^{-1} B_{\b\mu\nu}-  h^{-2}   \D^\rho  h   \QQ_{\rho\b \mu\nu}\\
  &=&h^{-1} B_{\b\mu\nu}-h^{-2} E^\rho  \QQ_{\rho\b \mu\nu}
  \eeaa
   Finally, recalling the definitions of $\SS$  and $B$,
  we deduce,
  \beaa
 \D^\a  \SS_{\a\b \mu\nu} &=&\D^\a \RR_{\a\b \mu\nu}+ 6 \D^\a \big(h^{-1}\QQ_{\a\b\mu\nu} \big)
 =6h^{-1} B_{\b\mu\nu}   - 6 h^{-2}    E^\rho  \QQ_{\rho\b \mu\nu}\\
 &=&-6h^{-1}\big[\FF_\b\,^\rho \T^\si\SS_{\si\rho\mu\nu}+(2/3)\II^\rho\,_{\b\mu\nu}\FF^{\ga\de}\T^\si\SS_{\si\rho \ga\de}\big]   - 6 h^{-2}   E^\rho     \QQ_{\rho\b \mu\nu}   \\
 &=&-6h^{-1}\T^\si\SS_{\si\rho\ga\de}\big[\FF_\b\,^\rho\, \de_\mu^\ga\, \de_\nu^\de+\frac 2 3 
\II^{\rho}\,_{\b\mu\nu}\FF^{\ga\de}\big] - 6 h^{-2}   E^\rho   \QQ_{\rho\b \mu\nu}
\eeaa

\end{proof}
\subsection{ Second Mars Tensor}
In \cite{Ma2} Mars   was able to give an alternative, stronger, characterization of the Kerr  solution in terms of         the vanishing  of the  tensor
 $\SS_{\a\b\mu\nu}\FF^{\mu\nu}$.  In what follows we show that a  simple modification of that tensor verifies a Maxwell equation.   We choose,
 \bea
 h=C(-\FF^2)^{1/4}
 \eea
 in the  generalized definition of $\SS$ in  \eqref{eq:gen-SS}, where $C$ is a constant      to be determined. With this choice of $h$ we define the complex self-dual 
 $2$-from,
\begin{equation}\label{ge30}
 \HH_{\a\b}:=h^{-3}\SS_{\a\b \mu\nu}\FF^{\mu\nu}.
 \end{equation}
Since $\QQ_{\a\b\mu\nu}\FF^{\mu\nu}=\frac 2 3 \FF^2 \FF_{\a\b}$ we also have,
\bea
\label{eq:forHH}
\HH_{\a\b}&=&h^{-3} \RR_{\a\b\mu\nu}\FF^{\mu\nu}- 4 C^{-4} \FF_{\a\b}.
\eea
\begin{proposition} 
\label{propp-Max}
The self-dual complex $2$-form $\HH_{\a\b}$ defined above
 verifies the  Maxwell equations,
 \begin{equation}\label{eq:Max}
\D^\a\HH_{\a\b}= -h^{-3}\T^\si(\SS\cdot\SS)_{\si\b} -3h^{-1}E^\rho\HH_{\rho\b}, \qquad
\end{equation}
where,
\beaa
(\SS\cdot\SS)_{\si\b}=\SS_\b\,^{\rho\mu\nu}\SS_{\si\rho \mu\nu } 
\eeaa
and,
\bea
 E_\rho :=\si_\rho +\D_\rho  h=-\frac 1 2  C^4 \T ^\si\HH_{\si\rho}
 \eea
 \end{proposition}
\begin{proof}
Recall   \eqref{ge4},
\beaa
\D_\a(\FF^2)&=&-4h^{-1}\sigma_\a\FF^2+2\T ^\la \SS_{\la\a\b\ga}\FF^{\b\ga}=      4h^{-1}\sigma_\a (h C^{-1})^4 +2\T ^\la \SS_{\la\a\b\ga}\FF^{\b\ga}\\
&=&4 C^{-4} h^3 \si_\a +2\T ^\la \SS_{\la\a\b\ga}\FF^{\b\ga}
\eeaa
Hence,
\beaa
C^{-1}D_\a h&=&  \D_\a\left[     (-\FF^2) ^{1/4}\right]=       - \frac 1 4 \D_\a\FF^2(-\FF^2)^{-3/4}=-\frac 1 4 [ h^4 C^{-4}]^{-3/4}  ( \D_\a\FF^2)\\
&= &-\frac 1 4 C^3  h^{-3} ( \D_\a\FF^2)   = -\frac 1 4 C^3 h^{-3} \big(4 C^{-4} h^3 \si_\a +2\T ^\la \SS_{\la\a\b\ga}\FF^{\b\ga}\big)\\
&=&-C^{-1} \si_\a - \frac 1 2 C^3    \T^\la \HH_{\la \a}
\eeaa
or,
\bea
\D_\b h&=&-\si_\b-  \frac 1 2 C^4 \T^\si\HH_{\si\b}. \label{formula-Dh}
\eea
We deduce,
\begin{equation}\label{ge31}
\begin{split}
E_\b
&=-\frac 1 2  C^4 \T ^\si\HH_{\si\b}
\end{split}
\end{equation}
We now calculate,
\beaa
\D^\a \HH_{\a \b}=\D^\a[\SS_{\a \b\mu\nu}\cdot h^{-3}\FF^{\mu\nu}]=\D^\a\SS_{\a \b\mu\nu}\cdot h^{-3}\FF^{\mu\nu}+\SS_{\a\b\mu\nu}\D^\a[h^{-3}\FF^{\mu\nu}].
\eeaa
Using proposition \ref{prop:BianchiS} ,
\beaa
\D^\a\SS_{\a\b\mu\nu}
&=&-6h^{-1}\T^\si\SS_{\si\rho\ga\de}\big[\FF_\b\,^\rho\, \de_\mu^\ga\, \de_\nu^\de+\frac 2 3 
\II^{\rho}\,_{\b\mu\nu}\FF^{\ga\de}\big] - 6 h^{-2}   E^\rho   \QQ_{\rho\b \mu\nu}
\eeaa
we deduce,
\beaa
\D^\a\SS_{a\b\mu\nu}\cdot h^{-3}\FF^{\mu\nu}&=&-6h^{-4}\FF^{\mu\nu}\big[\FF_\b\,^\rho \T^\si\SS_{\si\rho\mu\nu}+(2/3)\II^\rho\,_{\b\mu\nu}\FF^{\ga\de}\T^\si\SS_{\si\rho \ga\de}\big]  \\
&-&    6 h^{-2}   E^\rho   \QQ_{\rho\b \mu\nu} h^{-3} \FF^{\mu\nu}  \\
&=&-6h^{-4}\FF^{\mu\nu}\big[\FF_\b\,^\rho \T^\si\SS_{\si\rho\mu\nu}+(2/3)\II^\rho\,_{\b\mu\nu}\FF^{\ga\de}\T^\si\SS_{\si\rho \ga\de}\big]\\
&-&    6 h^{-5}   E^\rho \left( \FF_{\rho\b }\FF_{\mu\nu}-\frac{1}{3}\FF^2\II_{\rho\b \mu\nu} \right)  \FF^{\mu\nu}\\
  &=&-6h^{-1}\FF_\b\,^\rho \T^\si\HH_{\si\rho}+4h^{-1}\FF_{\b}\,^\rho \T^\si\HH_{\si\rho}   -4  h^{-5}(\FF^2)    E^\rho    \FF_{\rho\b }   \\
 &=&-2h^{-1} \FF_\b\,^\rho \T^\si\HH_{\si\rho} +4  h^{-2}(\FF^2)    E^\rho    \FF_{\b\rho } \\
 &=&-2h^{-1} \FF_\b\,^\rho \T^\si\HH_{\si\rho} + 4  h^{-5}  ( -h ^4/C^4)E^\rho    \FF_{\b\rho }\\
&=& -2h^{-1} \FF_\b\,^\rho \T^\si\HH_{\si\rho} +  4  h^{-1} C^{-4} E^\rho    \FF_{\b\rho }\\
&=& -2h^{-1} \FF_\b\,^\rho\left( \T^\si\HH_{\si\rho}+2 C^{-4}E_\rho\right)
\eeaa
Thus, in view of \eqref{ge31},
\beaa
D^\a\SS_{a\b\mu\nu}\cdot h^{-3}\FF^{\mu\nu}&=&0.
\eeaa
 On the other hand, recalling \eqref{eq:interm-H} 
 \beaa
\D_\rho [h^{-3}\FF_{\mu\nu }]&=&2h^{-4}\FF^2\T^\si\II_{\si\rho\mu\nu }-3h^{-4}  E_\rho \FF_{\mu\nu }+h^{-3}\T^\la \SS_{\la\rho\mu\nu}
\eeaa
we have,
\beaa
\SS_{\rho\b\mu\nu}\D^\rho[h^{-3}\FF^{\mu\nu}]&=&\SS^{\rho}\,_{\b} \,^{\mu\nu}\left[     2h^{-4}\FF^2\T^\si\II_{\si\rho\mu\nu }+   h^{-3} \T^\si\SS_{\si\rho \mu\nu }-3h^{-4}  E_\rho \FF_{\mu\nu }\right]
\eeaa
Observe that $\SS^\rho \,_\b\, ^{\mu\nu}\II_{\si\rho\mu\nu }=0$
 Thus,
\beaa
\SS_{\rho\b\mu\nu}\D^\rho[h^{-3}\FF^{\mu\nu}]&=&\SS^ \rho \,_\b\,^{\mu\nu}  h^{-3} \T ^\si\SS_{\si\rho \mu\nu }     - 3 h^{-1}\HH^\rho\,_\b  E_\rho     
\eeaa
Finally we deduce,
\beaa
\D^\a \HH_{\a \b}&=&\D^\a\SS_{\a \b\mu\nu}\cdot h^{-3}\FF^{\mu\nu}+\SS_{\a\b\mu\nu}\D^\a[h^{-3}\FF^{\mu\nu}]\\
&=&\SS^ \rho \,_\b\,^{\mu\nu}  h^{-3} \T ^\si\SS_{\si\rho \mu\nu }     - 3 h^{-1}\HH^\rho\,_\b  E_\rho \\
&=& -h^{-3}\T^\si(\SS\cdot\SS)_{\si\b} -3h^{-1}E^\rho\HH_{\rho\b}
\eeaa

\end{proof}
\end{appendix}

\end{document}